\documentclass{vldb}
\usepackage{amsmath}
\usepackage{amsfonts}
\usepackage{amssymb}
\usepackage{multirow}
\usepackage{hyperref}
\usepackage{algorithm}
\usepackage{algorithmic}
\usepackage{graphicx}
\usepackage{url}
\usepackage[usenames]{color}
\usepackage{colortbl}

\graphicspath{{figures/}}
\title{Bayesian Inference under Differential Privacy}

\numberofauthors{2}
\author{
\alignauthor
Yonghui Xiao\\
       \affaddr{Emory University}\\
       \affaddr{Atlanta, GA, USA}\\
       \email{yonghui.xiao@emory.edu}\\
\alignauthor
Li Xiong\\
       \affaddr{Emory University}\\
       \affaddr{Atlanta, GA, USA}\\
        \email{lxiong@emory.edu}
}

\begin{document}
\maketitle
\begin{abstract}
Bayesian inference is an important technique throughout statistics.
The essence of Beyesian inference is to derive the posterior belief updated from prior belief by the learned information, which is a
set of differentially private answers under differential privacy. Although Bayesian inference can be used in a variety of applications,
it 
becomes theoretically hard to solve when the number of differentially private answers is large.
To facilitate Bayesian inference under differential privacy,
this paper proposes a systematic mechanism. The key step of the mechanism is the implementation of Bayesian updating
with the best linear unbiased estimator derived by Gauss-Markov theorem. In addition, we also apply the proposed inference mechanism into
an online query-answering system, the novelty of which is that the utility for users is guaranteed by Bayesian inference in the form of
credible interval and confidence level. Theoretical and experimental analysis are shown to demonstrate the efficiency and effectiveness of
both inference mechanism and online query-answering system.

%
%
\end{abstract}

\newtheorem{theorem}{Theorem}[section]
\newtheorem{lemma}{Lemma}[section]
\newtheorem{definition}{Definition}[section]
\newtheorem{corollary}{Corollary}[section]
\newtheorem{observation}{Observation}[section]
\section{Introduction}
\label{sec-intro}
Data privacy issues frequently and increasingly arise for data sharing and data analysis tasks. Among all the
privacy-preserving mechanisms, differential privacy has been widely accepted for its strong privacy guarantee
~\cite{DBLP:conf/tamc/Dwork08,cacm}. It requires that the outcome of any computations or queries is formally
indistinguishable when run with and without any particular record in the dataset.
To achieve differential privacy the answer of a query is perturbed by a random noise whose magnitude is determined by a parameter, privacy budget.

Existing mechanisms\cite{Dwork-calibrating,Dwork-dp,Blum-SuLQ,McSherry-PINQ} of differential privacy only has
 a certain bound of privacy budget to spend on all queries.
We define the privacy bound as {\bf overall privacy budget} and such system {\bf bounded differentially private system}.
To answer a query, the mechanism allocates some privacy budget(also called privacy cost of the query) for it.
 Once the overall privacy budget is exhausted,
either the database has to be shut down or any further query would be rejected.
To prevent budget depletion and extend the lifetime of such systems, users have the burden to allocate privacy budget for the system.
However, there is no theoretical clue on how to allocate budget so far.

To save privacy budget, existing work uses the correlated answers to make inference about  new-coming queries.
Current literature falls into one of the following:
1. {\bf Query oriented strategy.} Given a set of queries(or answers) and the bound of privacy budget, optimize answers according to the correlation of queries\cite{Optimizing-PODS,Adaptive-query-strategy,Improving-utility-PCA};
2. {\bf Data oriented strategy.} Given the privacy bound, maximize the global utility of released data\cite{boost-accuracy,DBLP:conf/icde/XiaoWG10,SDMpaper};
3. {\bf Inference oriented strategy.} Use traditional inference method, like MLE, to achieve an inference result or bound\cite{Adam-MLE-estimator,McSherry-Probabilistic-inference,how-much-is-enough}.
All existing work can only make point estimation, which provides limited usefulness due to lack of probability properties.
Some work gives an error tolerance of the inference, like $[\epsilon,\delta]$-usefulness\cite{Blum-learningapproach}.
However, Bayesian inference has not been achieved yet.

Bayesian inference is an important technique throughout statistics.
The essence of Beyesian inference is to derive the posterior belief updated from prior belief by the learned information, which is a
set of differentially private answers in this paper. With Bayesian inference, a variety of applications can be tackled. Some examples are shown as follows.
\begin{itemize}
\begin{item}
{\bf hypothesis testing.} If a user(or an attacker) makes a hypothesis that Alice's income is higher than $50$ thousand dollars, given a set of Alice's noisy income
\footnote{It complies with the assumption of differential privacy that an adversary knows the income of all other people but Alice.}
, what is the probability that the hypothesis holds?
\end{item}
\begin{item}
{\bf credible interval and confidence level.} If a user requires an interval in which the true answer lies with confidence level $95\%$,
how can we derive such an interval satisfying the requirement?
\end{item}
\end{itemize}
Other applications, like regression analysis, statistical inference control and statistical decision making, can also be facilitated by Bayesian inference.

%
%

Although having so many applications,
Bayesian inference becomes theoretically hard to solve when the number of differentially private answers, denoted as ``history'' queries, is large.
This phenomenon is often referred as the curse of high dimensionality if we treat each answer in ``history'' as an independent dimension.
Because the derivation of posterior belief involves a series of integrals of probability function
\footnote{For convenience and clearance, probability function only has two meanings in this paper: probability density function for continuous variables and probability mass function for discrete variables.}
, we show later in this paper the complexity of the probability function
although it can be processed in a closed form\footnote
{A closed form expression can be defined by a finite number
of elementary functions(exponential, logarithm, constant, and nth root functions) under operators $+,-,\times,\div$.}.

{\bf Contributions.} This paper proposes a systematic mechanism to achieve Bayesian inference under differential privacy.
Current query-answering mechanisms, like Laplace mechanism\cite{Dwork-calibrating} and exponential mechanism\cite{McSherry-mechanism}, have also been incorporated in our approach.
According to these mechanisms, a set of ``history'' queries and answers with arbitrary noises can be given in advance.
The key step of our mechanism is the implementation of Bayesian updating about a new-coming query using the set of ``history'' queries and answers.
In our setting, uninformative prior belief is used, meaning we do not assume any evidential prior belief about the new query.
At first,
a BLUE(Best Linear Unbiased Estimator) can be derived
using Gauss-Markov theorem or Generalized Least Square method.
Then we propose two methods, Monte Carlo(MC) method and Probability Calculation(PC) method,
to approximate the probability function.
At last, the posterior belief can be derived by updating the prior belief using the probability function.
 Theoretical and experimental analysis have been given to show the efficiency and accuracy of two methods.

The proposed inference mechanism are also applied in an online utility driven query-answering system.
First,
it can help users specify the privacy budget by letting users demand the utility requirement in the form of credible interval and confidence level.
The key idea is to derive the credible interval and confidence level from the history queries using Bayesian inference.
If the derived answer satisfies user's requirement, then no budget needs to be allocated because the estimation can be returned.
Only when the estimation can not meet the utility requirement, the query mechanism is invoked for a differentially private answer.
In this way, not only the utility is guaranteed for users, but also the privacy budget can be saved so that the lifetime of the system can be extended.
Second, We further save privacy budget by allocating the only \textit{necessary} budget calculated by the utility requirement to a query.
Third, the overall privacy cost of a system is also measured to determine
whether the system can answer future queries or not. Experimental evaluation has been shown to demonstrate the utility and efficiency.

All the algorithms are implemented in MATLAB, and all the functions' names are consistent with MATLAB.

\section{Preliminaries and definitions}
\label{sec-preliminaries}

We use bold characters to denote vector or matrix, normal character to denote one row of the vector or matrix;
subscript $i,j$ to denote the $ith$ row, $jth$ column of the vector or matrix;
operator $[\cdot]$ to denote an element of a vector;
$\theta,\hat{\theta}$ to denote the true answer and estimated answer respectively;
$\mathcal{A}_Q$ to denote the differentially private answer of Laplace mechanism. 

%

\subsection{Differential privacy and Laplace mechanism}
\label{cha2:dp}
\begin{definition}[$\alpha$-Differential privacy \cite{Dwork-dp}]
\label{def-DP}
A data access mechanism $\mathcal{A}$ satisfies $\alpha$-differential privacy\footnote{Our definition is consistent with the unbounded model in\cite{Kifer-no-free-lunch}.}
if for any neighboring databases
 $D_1$ and $D_2$, for any query
function $Q$, $r\subseteq Range(Q)$\footnote{Range(Q) is the domain of the differentially private answer of Q.}, $\mathcal{A}_Q(D)$ is the mechanism to return an answer to query $Q(D)$,
\begin{equation}
\operatorname{sup}_{r\subseteq Range(Q)}\frac{Pr(\mathcal{A}_Q(D)=r | D=D_1)}{Pr(\mathcal{A}_Q(D)=r | D=D_2)}\leq e^{\alpha}
\end{equation}
$\alpha$ is also called privacy budget.
\end{definition}
Note that it's implicit in the definition that the privacy parameter $\alpha$ can be made public.
%

\begin{definition}[Sensitivity]
For arbitrary neighboring databases $D_1$ and $D_2$, the sensitivity of
a query $Q$ is the maximum difference between the query results of
$D_1$ and $D_2$,
\begin{displaymath}
S_Q=max|Q(D_1)-Q(D_2)|
\end{displaymath}
\end{definition}

For example, if one is interested in the population size of $0\sim30$ old, then we can pose this query:
\begin{itemize}
\item
$\textbf{Q}_1$: select count(*) from data where $0\leq age\leq 30$
\end{itemize}
$\textbf{Q}_1$ has sensitivity 1 because any change of 1 record can only affect the result by 1 at most.

\begin{definition}[Laplace mechanism]
{Laplace mechanism}\cite{Dwork-calibrating} is such a data access mechanism that given a user query $Q$ whose true answer is $\theta$, the returned answer $\tilde{\theta}$ is
\begin{align}
\label{fml-lap-mec}
\mathcal{A}_Q(D)=\theta+\tilde{N}(\alpha/S)
\end{align}
\end{definition}
$\tilde{N}(\alpha/S)$ is a random noise of Laplace distribution.
If $S=1$, the noise distribution is like equation (\ref{fml-Laplace-distribution}).
\begin{align}
\label{fml-Laplace-distribution}
f(x,\alpha)=\frac{\alpha}{2}exp(-{\alpha |x|});
\end{align}%
For $S\neq 1$, replace $\alpha$ with $\alpha/S$ in equation (\ref{fml-Laplace-distribution}).
\begin{theorem}[\cite{Dwork-calibrating,Dwork-dp}]
\label{theo-Lap-DP}
Laplace mechanism achieves $\alpha$-differential privacy, meaning that adding $\tilde{N}(\alpha/S_Q)$ to $Q(D)$ guarantees $\alpha$-differential privacy.
\end{theorem}

%
%
%
%

\subsection{Utility}
\label{cha2:utility}

For a given query $Q$, denote $\theta$ the true answer of $Q$.
By some inference methods, a point estimation $\hat{\theta}$ can be obtained by some metrics, like Mean Square Error(MSE) $\operatorname{E}[(\hat{\theta}-\theta)^2]$.
Different from point estimation, interval estimation specifies instead a range within which the answer lies.
We introduce $(\epsilon,\delta$)-usefulness\cite{Blum-learningapproach} first, then show that it's actually a special case of credible interval, which will be used in this paper.

\begin{definition}[$(\epsilon,\delta$)-usefulness \cite{Blum-learningapproach}]
A query answering mechanism is $(\epsilon,\delta)$-useful for
query $Q$ if
$Pr(|\hat{\theta}-\theta|\leq\epsilon)\geq 1-\delta$.
\end{definition}
\begin{definition}[credible interval\cite{book-statistical-inference}]
\label{definition-usefulness}
Credible interval with confidence level $1-\delta$ is a range $[L,U]$
with the property:
\begin{align}
Pr(L\leq \theta \leq U)\geq 1-\delta
\end{align}
\end{definition}

In this paper, we let $2\epsilon=U-L$ to denote the length of credible interval.
Note when $\hat{\theta}$ is the midpoint of $L(\tilde{\Theta})$ and $U(\tilde{\Theta})$, these two definitions are the same and have equal $\epsilon$.
Thus $(\epsilon,\delta)$-usefulness is a special case of credible interval.
An advantage of credible interval is that users can specify certain intervals to calculate confidence level. For example, a user may be interested in the probability that $\theta$ is larger than $10$.

In our online system, the utility requirement is the
$1-\delta$ credible interval whose parameter $\delta$ is specified by a user to demand the confidence level of returned answer.
Intuitively, the narrower the interval, the more useful the answer. Therefore, we also let users specify the parameter $\epsilon$ so that
the length of interval can not be larger than $2\epsilon$, $U-L\leq 2\epsilon$.

\subsection{Bayesian Updating}

The essence of Beyesian inference is to update the prior belief by the learned information(observations),
which is actually a set of differentially private answers in this paper.
Assume a set of observations is given as $\tilde{\theta}_1,\tilde{\theta}_2,\cdots,\tilde{\theta}_n$. In step $i$,
denoted $f_{i-1}(\theta)$ the prior belief of $\theta$, which can be updated using the observation $\tilde{\theta_i}$ by Bayes' law as follows:
\begin{align}
\label{fml-updating}
f_i(\theta)=f_{i-1}(\theta)\frac{Pr(\tilde{\theta}_i=\theta|\theta)}{Pr(\tilde{\theta}_i)}
\end{align}
where $f_i(\theta)$ is the posterior belief at step $i$. It is also the prior belief at step $i+1$.

Above equation justifies why we use uninformative prior belief(also called a priori probability\cite{book-statistical-inference} in some literature) as the original prior belief
because any informative(evidential) prior belief can be updated from uninformative belief. It is the same we take either the informative prior belief
or the posterior belief updated from uninformative prior belief as the prior belief.
%

\subsection{Data model and example}
\label{cha3:sec-model-intro}
%

\subsubsection{Data}%

Consider a dataset with $N$ nominal or discretized attributes, we use an $N$-dimensional data cube,
also called a base cuboid in the data warehousing literature \cite{JWH-dataming,Ding:2011:DPD:1989323.1989347},
to represent the aggregate information of the data set.  The records are the points in the $N$-dimensional data space.
Each cell of a data cube represents an aggregated measure, in our case, the count of the data points corresponding to
the multidimensional coordinates of the cell. We denote the number of cells by $n$ and
$n = |dom(A_1)| * \dots * |dom(A_N)|$ where $|dom(A_i)|$ is the domain size of attribute $A_i$.

\begin{figure}[h!]
\centering
\hspace{-0.8cm}
\includegraphics[width=10cm]{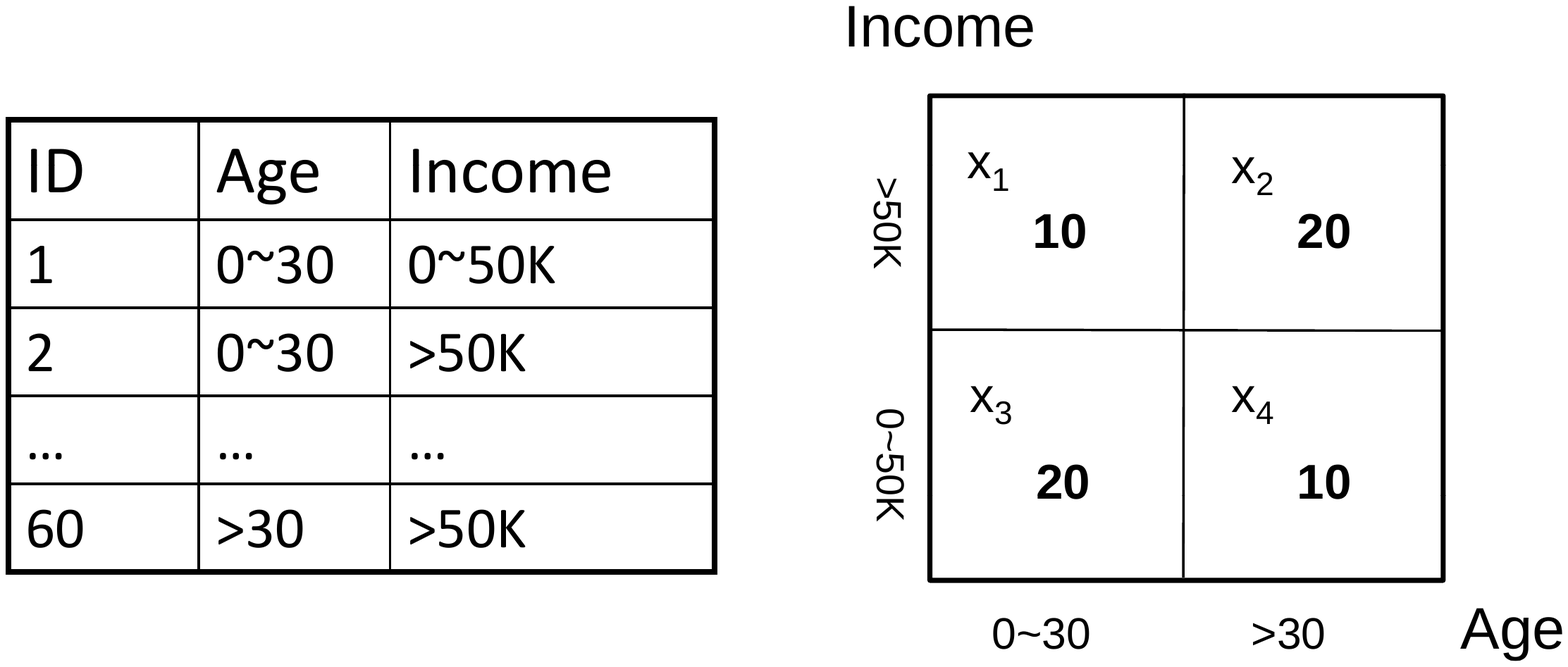}
\vspace{-4cm}
\caption{\small Example original data represented in a relational table (left) and a 2-dimensional count cube (right)}
\label{Fig_exampledata}
\end{figure}

\begin{itemize}
\item
{\bf Example.} Figure \ref{Fig_exampledata} shows an example relational dataset with attribute age and income (left) and a two-dimensional
count data cube or histogram (right). The domain values of age are $0\scriptsize{\sim}30$ and $>30$;
the domain values of income are $0\scriptsize{\sim}50K$ and $>50K$.
Each cell in the data cube represents the population count corresponding to the age and income values.
We can represent the original data cube, e.g. the counts of all cells, by an $n$-dimensional column vector $\textbf{x}$ shown below.

\begin{small}
\begin{equation}
\label{eqn_x}
\textbf{x}=
\left[
\begin{array}{ccccccccc}
10&20&20&10
\end{array}
\right]^T
\end{equation}
\end{small}
\end{itemize}
\subsubsection{Query}
We consider linear queries that compute a linear combination of the count values in the data cube based on a query predicate.

\begin{definition}[Linear query \cite{Optimizing-PODS}] A set of linear queries $\textbf{Q}$ can be represented as an $q\times n$-dimensional matrix
$\textbf{Q} = [\textbf{Q}_1; \dots ;Q_q]$ with each $Q_i$ is a coefficient vector of $\textbf{x}$.
The answer to $\textbf{Q}_i$ on data vector $\textbf{x}$ is the product
$\textbf{Q}_i\textbf{x}$ = $\textbf{Q}_{i1}\textbf{x}_1 +\textbf{Q}_{i2}\textbf{x}_2+ \dots + \textbf{Q}_{in}\textbf{x}_n$.
\end{definition}
\begin{itemize}
\item
{\bf Example.} If we know people with $income\leq 50K$ don't pay tax, people with $income>50K$ pay 2K dollars,
and from the total tax revenue 10K dollars are invested for education each year,  then
to calculate the total tax, following query $\textbf{Q}_2$ can be issued.
\begin{itemize}
\item
$\textbf{Q}_2$:    select 2$*$count() where $income>50$
    $\ \ -$ 10
\end{itemize}
\begin{align}
\label{fml-query-Q2}
\begin{small}
\textbf{Q}_2=\left[
\begin{array}{cccc}
2 &2&0&0
\end{array}
\right]
\end{small}
\end{align}
\vspace{-0.5cm}
\begin{align}
\label{fml-AQ2}
\mathcal{A}_{Q_2}=\textbf{Q}_2 \cdot \textbf{x}+\tilde{N}(\alpha/S_{Q_2})
\end{align}
$\textbf{Q}_2$ has sensitivity $2$ because any change of $1$ record can affect the result by $2$ at most.
Note that the constant $10$ does not affect the sensitivity of $\textbf{Q}_2$ because it won't vary for any record change
\footnote{So we ignore any constants in $\textbf{Q}$ in this paper.}.
Equation (\ref{fml-query-Q2}) shows the query vector of $\textbf{Q}_2$.
The returned answer $\mathcal{A}_{Q_2}$
 from Laplace mechanism are shown in equation (\ref{fml-AQ2}) where $\alpha$ is the privacy cost of $\textbf{Q}_2$.
\end{itemize}

%
Given a $m \times n$ query ``history'' \textbf{H}, the query answer for \textbf{H} is a length-$m$ column vector of query results, which can be computed as the matrix product \textbf{Hx}.



%

\subsubsection{Query history and Laplace mechanism}
%
For different queries, sensitivity of each query vector may not be necessarily the same.
Therefore, we need to compute the sensitivity of
each query  by Lemma \ref{lemma-GS}.

\begin{lemma}
\label{lemma-GS}
For a linear query  $\textbf{Q}$, the sensitivity $S_{\textbf{Q}}$ is $max(abs(\textbf{Q}))$.
\end{lemma}
where function ``$abs()$'' to denote the absolute value of a vector or matrix, meaning that for all $j$, $[abs(\textbf{Q})]_j=|\textbf{Q}_j|$.
$max(abs(\textbf{Q}))$ means the maximal value of $abs(\textbf{Q})$.

We can write the above equations in matrix form.
Let $\textbf{y}$, $\textbf{H}$, $\boldsymbol\alpha$ and $\textbf{S}$ be the matrix form of
$\mathcal{A}_Q$, $\textbf{Q}$, $\alpha$ and $S$ respectively. According to Lemma \ref{lemma-GS}, $\textbf{S}=max_{row}(abs(\textbf{H}))$
where $max_{row}(abs(\textbf{H}))$ is a column vector whose items are the maximal values of each row in $abs(\textbf{H})$.

Equations(\ref{fml-y}, \ref{fml-SH}) summarize the relationship of $\textbf{y}$, $\textbf{H}$, $\textbf{x}$, $\tilde{\textbf{N}}$ and $\boldsymbol\alpha$.
\begin{align}
\label{fml-y}
\textbf{y}=\textbf{Hx}+\tilde{\textbf{N}}(\boldsymbol\alpha./\textbf{S})
\\
\label{fml-SH}
\textbf{S}=max_{row}(\textbf{abs(\textbf{H})})
\end{align}
where
operator ./ means $\forall i$,
$[\boldsymbol\alpha./\textbf{S}]_i=\boldsymbol\alpha_i/S_i$.
\begin{itemize}
\item
{\bf Example.} Suppose we have a query history matrix $\textbf{H}$ as equation (\ref{fml-H}), $\textbf{x}$ as equation (\ref{eqn_x})
and the privacy parameter $\textbf{A}$ used in all query vector shown in equation (\ref{fml-A}).

\begin{align}
\small
\label{fml-H}
\textbf{H}=
\left[
\begin{tabular}{c
ccc
}
1 & 1& 0 &0  \\
0& 0 &1 &1\\
0& 0 &0 &1 \\
0 &0 &1 &0 \\
0 &1 &0 &1\\
2 &1& 0& 0 \\
0& 0 &2 &-1\\
0 &-1 &0 &1 \\
\end{tabular}
\right]
\end{align}
\begin{align}
\label{fml-A}
\begin{small}
\boldsymbol\alpha=
\left[
\begin{tabular}{cccccccc}
.05&.1&.05&.1&.1&.05&.05& .1
\end{tabular}
\right]^T
\end{small}
\end{align}

Then according to Lemma \ref{lemma-GS}, $\textbf{S}$ can be derived in equation (\ref{fml-S}).
\begin{align}
\label{fml-S}
\textbf{S}=\left[
\begin{array}{ccccccccc}
1& 1& 1& 1& 1& 2& 2&1
\end{array}
\right]^T
\end{align}

Therefore, the relationship of $\textbf{y}$, $\textbf{H}$, $\textbf{A}$ and $\textbf{S}$ in equation (\ref{fml-y}) can be illustrated in equation (\ref{fml-HI-color})
\begin{align}
\tiny
\label{fml-HI-color}
\left[
\begin{array}{c}
\textbf{y}_1\\
\textbf{y}_2\\
\textbf{y}_3\\
\textbf{y}_4\\
\textbf{y}_5\\
\textbf{y}_6\\
\textbf{y}_7\\
\textbf{y}_8\\
\end{array}
\right]
=
\left[
\begin{tabular}{c
ccc
}
1 & 1& 0 &0  \\
0& 0 &1 &1\\
0& 0 &0 &1 \\
0 &0 &1 &0 \\
0 &1 &0 &1\\
2 &1& 0& 0 \\
0& 0 &2 &-1\\
0 &-1 &0 &1 \\
\end{tabular}
\right]
\textbf{x}
+
\tilde{\tiny{N}}
\left(
\begin{tabular}{l}
0.05/1\\
    0.1/1\\
    0.05/1\\
    0.1/1\\
    0.1/1\\
    0.05/2\\
    0.05/2\\
    0.1/1\\
\end{tabular}
\right)
\end{align}
\end{itemize}

\section{the Best Linear Unbiased Estimator}
\label{sec-online}
Given a set of target queries as $\textbf{Q}$ and equations (\ref{fml-y},\ref{fml-SH}), the BLUE can be derived in the following equation by Gauss-Markov theorem or Generalized Least Square method so that
 $\operatorname{MSE}(\hat{\textbf{x}})$ can be minimized.
%
\begin{align}
\label{eqn-x-hat}
\hat{\textbf{x}}=(\textbf{H}^Tdiag^2(\boldsymbol\alpha./\textbf{S})\textbf{H})^{-1}\textbf{H}^Tdiag^2(\boldsymbol\alpha./\textbf{S})\textbf{y}
\end{align}
where function $diag()$ transforms a vector to a matrix with each of
element in diagonal; $diag^2()=diag()*diag()$.
%
For estimatability, we assume $rank(\textbf{H})=n$. Otherwise $\textbf{x}$ may not be estimable
\footnote{However, $\Theta=\textbf{Qx}$ may also be estimable
iff $\textbf{Q}=\textbf{Q}(\textbf{H}^T\textbf{H})^-\textbf{H}^T\textbf{H}$ where $(\textbf{H}^T\textbf{H})^-$
is the generalized inverse of $\textbf{H}^T\textbf{H}$.}.

\begin{theorem}
\label{theo-hat-theta-BLUE}
To estimate $\Theta=\textbf{Qx}$, which is a $q\times 1$ query vector, $\textbf{Q}\in \mathbb{R}^{q\times n}$, the linear
estimator $\hat{\Theta}=\textbf{Ay}$ achieves minimal MSE when
\begin{align}
\textbf{A}=\textbf{Q}(\textbf{H}^Tdiag^2(\boldsymbol\alpha./\textbf{S})\textbf{H})^{-1}\textbf{H}^Tdiag^2(\boldsymbol\alpha./\textbf{S})
\end{align}
\end{theorem}

In the rest of the paper, we will focus on a query. Thus we assume $\textbf{Q}\in \mathbb{R}^{1\times n}$ and $\theta=\textbf{Qx}$.
However, this work can be easily extended to
multivariate case using multivariate Laplace distribution\cite{ML-distribution}.
\begin{corollary}
\label{corollary-theta-unbiased}
$\hat{\theta}$ is unbiased, meaning $\operatorname{E}(\hat{\theta})=\theta$.
\end{corollary}
\begin{corollary}
\label{corollary-var-theta-i}
The mean square error of $\hat{\theta}$ equals to variance of $\hat{\theta}$,
\begin{align}
\operatorname{MSE}(\hat{\theta})=\operatorname{Var}(\hat{\theta})=2\textbf{Q}(\textbf{H}^Tdiag^2(\boldsymbol\alpha./\textbf{S})\textbf{H})^{-1}\textbf{Q}^T
\end{align}
\end{corollary}

A quick conclusion about $(\epsilon,\delta)$-usefulness can be drawn using Chebysheve's inequality.
\begin{theorem}
The $(\epsilon,\delta)$-usefulness for $\theta$ is satisfied when
$\delta=\frac{Var(\hat{\theta})}{\epsilon^2}$, which means
\begin{align}
Pr(|\theta-\hat{\theta}|\leq \epsilon)\geq 1-\frac{Var(\hat{\theta})}{\epsilon^2}
\end{align}
\end{theorem}

However, the above theorem only gives the bound of $\delta$ instead of the probability $Pr(|\Theta_i-\hat{\Theta}_i|\leq \epsilon)$ which can
be derived  by   Bayesian inference.
\begin{itemize}
\item
{\bf Example.}
Let $
\textbf{Q}=\left[
\begin{array}{ccccccccc}
1 &0 &1 &0
\end{array}
\right]
$,
$
\textbf{y}=[
   30.8,\
   30.3,\
   46.9,\\\
   20.2,\
   30.4,\
   68.9,\
   38.9,\
    9.5
  ]^T
$.
First, we can derive $diag(\boldsymbol\alpha./\textbf{S})$ as
\begin{align*}
\small
diag(\boldsymbol\alpha./\textbf{S})=
\left[
\begin{array}{cccccccc}
.05      &   0      &   0    &     0     &    0     &    0       &  0      &   0    \\
         0   & .1    &     0    &     0    &     0   &      0   &      0   &      0  \\
         0     &    0  &  .05     &    0&         0        & 0&         0        & 0\\
         0     &    0 &        0   & .1  &       0        & 0  &       0        & 0  \\
         0     &    0  &       0     &    0 &   .1       &  0   &      0       &  0   \\
         0      &   0   &      0    &     0  &       0    &.025    &     0      &   0    \\
         0     &    0    &     0   &      0   &      0   &      0   & .025    &     0    \\
         0     &    0     &    0  &       0    &     0  &       0    &     0  &  .1      \\
\end{array}\right]
\end{align*}

Then by Theorem \ref{theo-hat-theta-BLUE}, we have
$\hat{\textbf{x}}=$
$
\left[
24.9,
   10.1,
   17.0,
   19.5
\right]^T
$,
$\hat{\theta}=42.0$.
\end{itemize}

\section{Bayesian Inference}
\label{sec-Bayesian-inference}
In equation (\ref{fml-updating}), Bayesian inference involves three major components: prior belief, observation(s) and the conditional probability
$Pr(\tilde{\theta}_i|\theta)$.
Then from last section, an estimator $\hat{\theta}=\textbf{Ay}$ can be obtained.
We take uninformative prior belief, meaning we do not assume any evidential prior belief about $\theta$.
Because $Pr(\tilde{\theta}_i)$ can be calculated as $\int f(\tilde{\theta}_i|\theta)f_{i-1}(\theta)d\theta$,
the remaining question is to calculate the conditional probability $Pr(\tilde{\theta}_i=\theta|\theta)$.

We denote $f_Z(z)$ a probability function of $Z$, $f_Z(x)$ a probability function of $Z$ with input variable $x$.
For example, if $f_{\tilde{N}}(x)$ is the probability function of $\tilde{N}$ whose probability function is shown in equation (\ref{fml-Laplace-distribution}),
following equation shows $f_{\tilde{N}}(z)$.
\begin{align}
\label{fml-Laplace-distribution-z}
f_{\tilde{N}}(z)=\frac{\alpha}{2}exp(-{\alpha |z|});
\end{align}%
Then $Pr(\tilde{\theta}_i=\theta|\theta)=f_{\theta}(z)|_{\theta=z}$.

\subsection{Theoretical Probability Function}
\begin{lemma}
\label{lemma-pdf-theta-i}
The probability function of $f_{\theta}(\theta)$ is the probability function of $f_{\textbf{A}\tilde{\textbf{N}}}(\textbf{Ay}-\theta)$.
\end{lemma}
\begin{proof}
\begin{align*}
\hat{\theta}=\textbf{Ay}=\textbf{A}(\textbf{Hx}+\tilde{\textbf{N}})=\textbf{Qx}+\textbf{A}\tilde{\textbf{N}}=\theta+\textbf{A}\tilde{\textbf{N}}\\
\Longrightarrow \theta=\textbf{Ay}-\textbf{A}\tilde{\textbf{N}}\\
\Longrightarrow Pr(\theta=\theta)=Pr[\textbf{A}\tilde{\textbf{N}}=\textbf{Ay}-\theta]
\end{align*}
Because of the above equation, we have
\begin{align*}
f_{\theta}(\theta)=f_{\textbf{A}\tilde{\textbf{N}}}(\textbf{Ay}-\theta)
\end{align*}
\end{proof}

\begin{theorem}
The probability function of $\theta$ is
\begin{align}
\label{eqn-pdf-theta-i}
f_{\theta}(\theta)=\frac{1}{2\pi}\int_{-\infty}^{\infty}exp(-it(\textbf{Ay}-\theta))\prod_{k=1}^m \frac{{\boldsymbol\alpha_k^2}}{{\boldsymbol\alpha_k^2}+(\textbf{A}_{k}^2\textbf{S}_k^2)t^2}dt
\end{align}
\end{theorem}
\begin{proof}
The noise for a query with sensitivity $\textbf{S}_i$ is Lap$(\boldsymbol\alpha_i/\textbf{S}_i)$.
It's easy to compute the characteristic function of Laplace distribution is
\begin{align}
\Phi_{\tilde{\textbf{N}}_i}(t)=\frac{\boldsymbol\alpha_i^2}{\textbf{S}_i^2t^2+\boldsymbol\alpha_i^2}
\end{align}
$\textbf{A}\tilde{\textbf{N}}=\sum_{k=1}^m\textbf{A}_{k}\tilde{\textbf{N}}_k$, so the characteristic function of $\textbf{A}\tilde{\textbf{N}}$ is
\begin{align}
\Phi_{\textbf{A}\tilde{\textbf{N}}}=\prod_{k=1}^{m}\frac{\boldsymbol\alpha_k^2}{\textbf{S}_k^2\textbf{A}_{k}^2t^2+\boldsymbol\alpha_i^2}
\end{align}
Then the probability function of $\textbf{A}\tilde{\textbf{N}}$ is
\begin{align*}
f_{\textbf{A}\tilde{\textbf{N}}}(\theta)=\frac{1}{2\pi}\int_{-\infty}^{\infty}exp(-it\theta)\prod_{k=1}^m \frac{{\boldsymbol\alpha_k^2}}{{\boldsymbol\alpha_k^2}+(\textbf{A}_{k}^2\textbf{S}_k^2)t^2}dt
\end{align*}
Therefore, by Lemma \ref{lemma-pdf-theta-i}, we have equation (\ref{eqn-pdf-theta-i}).
\end{proof}

The probability function can also be represented in a closed form using multivariate Laplace distribution in \cite{ML-distribution}.
It has also been proven in \cite{ML-distribution} that linear combination of Laplace distribution is also multivariate Laplace distributed.
To reduce the complexity, the probability function of sum of $n$ identical independent distributed Laplace noises
can be derived as follows\cite{Bilateral-gamma}.
\begin{small}
\begin{align}
\label{fml-baliteral-gamma}
f_n(z)=\frac{\alpha^{n}}{2^n  \Gamma^2(n)}exp(-{\alpha |z|})\int_0^{\infty}v^{n-1}(|z|+\frac{v}{2\alpha})^{n-1}e^{-v}dv
\end{align}
\end{small}

We can see that even the above simplified equation becomes theoretically hard to use when $n$ is large.
More probability functions can be found in \cite{book-Laplace}.

To circumvent the theoretical difficulty, two algorithms, Monte Carlo(MC) and discrete probability calculation(PC) method,
are proposed to approximate the probability function. Error of the approximation is defined to measure the accuracy of the two algorithms.
By comparing the error and computational complexity, we discuss that both of them have advantages and disadvantages.

\subsection{Monte Carlo Probability Function}
\label{sec-Monte-Carlo-Probability}
We show the algorithm to derive Monte Carlo probability as follows. It uses sampling technique to draw random variables from a
certain probability function. We skip detailed sampling technique, which can be found in literature of statistics,
like \cite{book-statistical-inference,book-probability-models,book-Laplace}.

\begin{algorithm}[H]
\caption{Derive the probability function of $\theta$ by Monte Carlo method}
\begin{algorithmic}
\label{alg-MC}
\REQUIRE{$m_s$: sample size for a random variable; $\boldsymbol\alpha$: privacy budget; $\textbf{S}$: sensitivity; $\textbf{A}$: $\theta=\textbf{Ay}$; $m$: $\textbf{H}\in \mathbb{R}^{m\times n}$}\\
1. $\textbf{Y}=\textbf{0}$, which is a ${m_s\times 1}$ vector;
\FOR{$k=1;k<m;k++$}
\STATE Draw $m_s$ random variables from distribution $Lap(\boldsymbol\alpha_k/\textbf{S}_k)$; denote the variables as $\textbf{X}_k$;\\
multiply $\textbf{A}_{k}$ to $\textbf{X}_k$, denote as $\textbf{A}_{k}\textbf{X}_k$;\\
$\textbf{Y}= \textbf{Y}+\textbf{A}_{k}\textbf{X}_k$;\\
\ENDFOR\\
2. $\textbf{Y}=round(\textbf{Y})$.\\
3. Plot the histogram of $\textbf{Y}$:\\
\hspace{4mm}$3.1\ $$|\textbf{u}|=2*max(abs(\textbf{Y}))+1$.\\
\hspace{4mm}$3.2\ $$\textbf{u}=hist(\textbf{Y},-\frac{|\textbf{u}|-1}{2}:1:\frac{|\textbf{u}|-1}{2})/m_s$.\\
\RETURN{The probability mass vector $\textbf{u}$}
\end{algorithmic}
\end{algorithm}
The function $round(\textbf{Y})$ rounds the elements of $\textbf{Y}$ to the nearest integers.
The purpose of step 3 is to to guarantee $|\textbf{u}|$ is an odd number. In step 4, function $hist(\textbf{Y},-\frac{|\textbf{u}|-1}{2}:1:\frac{|\textbf{u}|-1}{2})$ returns a
vector of number representing the frequency of the vector $-\frac{|\textbf{u}|-1}{2}:1:\frac{|\textbf{u}|-1}{2}$.
For example, $\textbf{u}[1]=\sum_{i=1}^{m_s} b(\textbf{Y}_i=-\frac{|\textbf{u}|-1}{2})/m_s$; $\textbf{u}[i]=\sum_{i=1}^{m_s} b(\textbf{Y}_i=-\frac{|\textbf{u}|-1}{2}+i)/m_s$
where $i\in \mathbb{Z}$ and $b()$ is a bool function that returns $1$ as true, $0$ as false.

The returned vector  $\textbf{u}$ is a probability mass function in a discretized domain.
Following definition explains a probability mass vector $\textbf{v}$.
\begin{definition}
\label{def-probability-mass-vector}
a probability mass vector $\textbf{v}_z$ is the probability mass function of $z$ so that
\begin{align*}
\textbf{v}_i=\int_{i-\frac{|\textbf{v}|}{2}-1}^{i-\frac{|\textbf{v}|}{2}}f_z(z)dz=F_z(i-\frac{|\textbf{v}|}{2})-F_z(i-\frac{|\textbf{v}|}{2}-1)
\end{align*}
where $|\textbf{v}|$ is the length of $\textbf{v}$, $f_z(z)$ is the probability density function of $z$, $F_z()$ is the cumulative distribution function.
\end{definition}
Note that $|\textbf{v}|$ should be an odd number implicitly to guarantee the symmetry of probability mass vector.

For example, if $\textbf{v}=[0.3, 0.4, 0.3]$, it means $\textbf{v}_1=\int_{-1.5}^{-0.5}f_z(z)dz=0.3$, $\textbf{v}_2=\int_{-0.5}^{0.5}f_z(z)dz=0.4$,
$\textbf{v}_3=\int_{0.5}^{1.5}f_z(z)dz=0.3$.

Following theorem can be proven with Central Limit Theorem.
\begin{theorem}
\label{theo-PDF-n}
$\textbf{u}$ in Algorithm \ref{alg-MC} converges in probability to $\textbf{A}\tilde{\textbf{N}}$, which means $\forall \delta>0$
\begin{align}
lim_{m_s\rightarrow \infty}Pr(|\textbf{u}_i-Pr(i-\frac{|\textbf{u}|}{2}-1\leq \textbf{A}\tilde{\textbf{N}}<i+\frac{|\textbf{u}|}{2})|<\delta)=1
\end{align}
\end{theorem}

%
%
%

At last, by Lemma \ref{lemma-pdf-theta-i}, we can derive the probability of $\theta$ by the following equation.
\begin{align}
\label{eqn-probability-mass-vector}
\textbf{u}_i=Pr(\textbf{Ay}+\frac{|\textbf{u}|}{2}-i < \theta \leq \textbf{Ay}+\frac{|\textbf{u}|}{2}-i+1 )
\end{align}

\begin{itemize}
\item
{\bf Example.}
First we can calculate $\textbf{A}$.
\begin{align}
\label{eqn-example-A}
\small
\textbf{A}=\left[
\begin{array}{cccccccc}
0.48   & 0.36  & -0.03 &   0.50 &  -0.50  &  0.26   & 0.07    &0.24
\end{array}
\right]
\end{align}


Let Sample size $m_s=10^6$. First, generate $10^6$ random variables from Laplace distribution $Lap(0.05)$.
Denote the random variables as the vector $\textbf{X}_1$.
Then generate $10^6$ random variables from $Lap(0.1)$, denoted as $\textbf{X}_2$. Let $\textbf{Y}=0.48\textbf{X}_1+0.36\textbf{X}_2$. Similarly, generate all the
random variables from Laplace distributions $Lap(\boldsymbol\alpha./\textbf{S})$. Finally we have a vector $\textbf{Y}=\sum(\textbf{A}_k.*\textbf{X}_k)$.
Next we round $\textbf{Y}$ to the nearest integers. $max(\textbf{Y})=466$, $\min(\textbf{Y})=-465$. It means the range of $\textbf{Y}$ is in
$[-465,\ 466]$. So let $|\textbf{u}|=2*466+1=933$. At last we make a histogram of $\textbf{Y}$ in the range $[-466,\ 466]$.
As in Equation ($\ref{eqn-probability-mass-vector}$), It represents the probability mass vector of $\theta$. For example,
$\textbf{u}[465]$ is $Pr(42-1.5\leq\theta<42-0.5)$. Figure \ref{Fig-probability-mass-vectors-MC} shows the the probability mass function of $\theta$.
\end{itemize}

\begin{figure}[h!]
\begin{minipage}{0.45\textwidth}
\centering
\includegraphics[width=8cm]{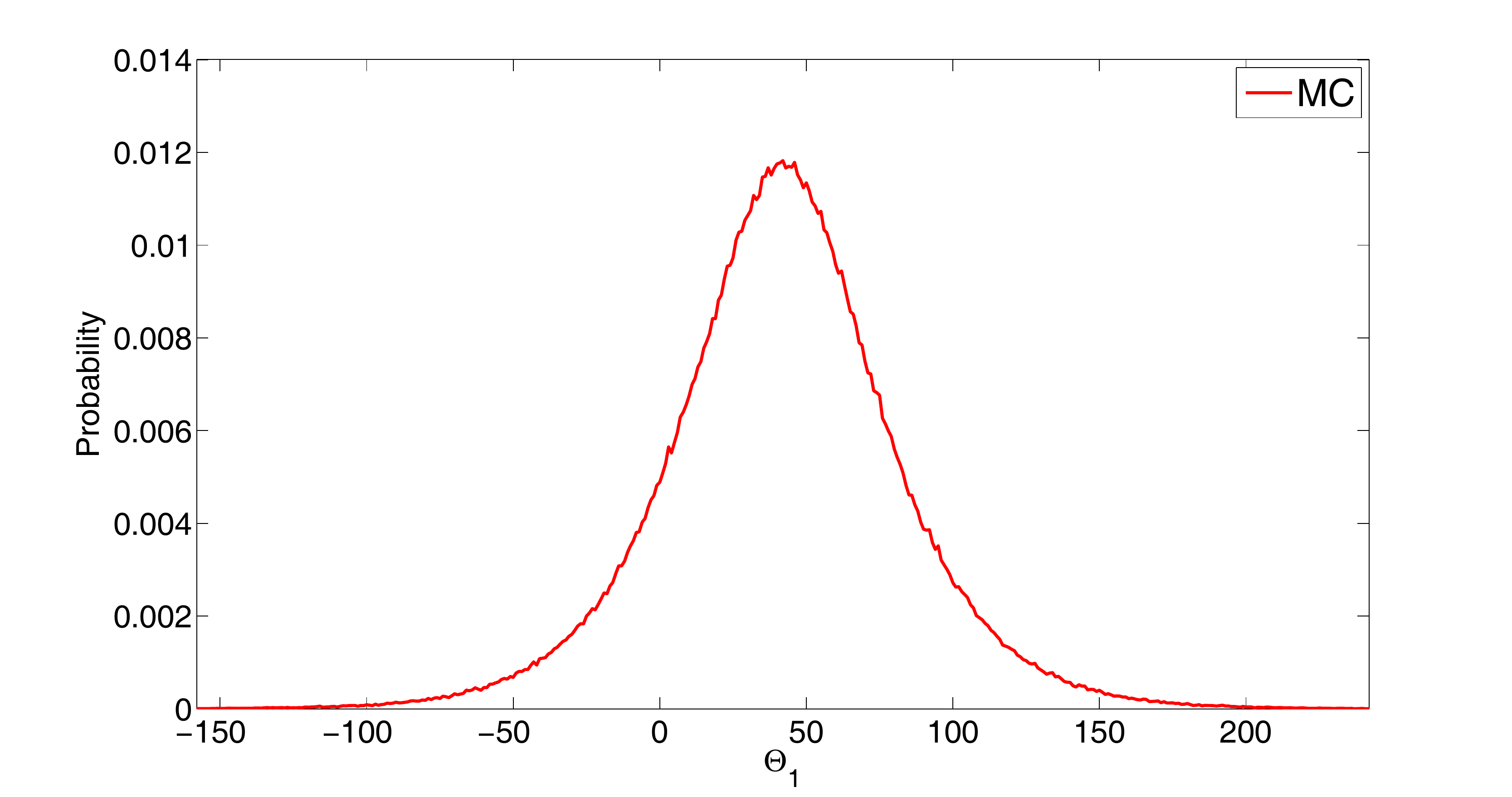}
\end{minipage}
\caption{Probability mass function of $\theta$.}
\label{Fig-probability-mass-vectors-MC}
\end{figure}

\subsubsection{Error Analysis}
We use the sum of variance to  measure the quality of the probability mass vector.
\begin{definition}
\label{def-error-u}
For a probability mass vector $\textbf{u}$,  the error of $\textbf{u}$ is defined as:
\begin{align}
error(\textbf{u})=\sum_{\textbf{u}_i} (\textbf{u}_i-\operatorname{E}({\textbf{u}}_i))^2
\end{align}
\end{definition}

Because each $\textbf{u}_i$ is a representation of
$Y_i=b(\textbf{Ay}+\frac{|\textbf{u}|}{2}-i < \theta \leq \textbf{Ay}+\frac{|\textbf{u}|}{2}-i+1)/m_s$, $error(\textbf{u})$ is actually $(m_s-1)\sum \sigma_{Y_i}^2$ where
$\sigma_{Y_i}^2$ is the variance of $Y_i$.

\begin{theorem}
\label{theo-error-monte-carlo}
The error of $\textbf{u}$ in Algorithm \ref{alg-MC} satisfies
\begin{align*}
error(\textbf{u})<\frac{\sigma^2}{m_s}z
\end{align*}
where $ z\sim \chi^2(|\textbf{u}|)$ and $\sigma^2$ is the maximum variance among $\textbf{u}$.
\end{theorem}
\begin{proof}
\begin{align*}
error(\textbf{u})=\sum_{\textbf{u}_i\in \textbf{u}} (\textbf{u}_i-\operatorname{E}(\textbf{u}_i))^2\\
=\frac{\sigma^2}{m_s}\sum_{\textbf{u}_i\in \textbf{u}}[\frac{\sqrt{m_s}}{\sigma}(\textbf{u}_i-\operatorname{E}(\textbf{u}_i))]^2 \ where \ (\sigma=max(\sigma_i))\\
<\frac{\sigma^2}{m_s}\sum_{\textbf{u}_i\in \textbf{u}}[\frac{\sqrt{m_s}}{\sigma_i}(\textbf{u}_i-\operatorname{E}(\textbf{u}_i))]^2
\end{align*}
In $\textbf{u}$, each $\textbf{u}_i$ represents a value derived by Monte Carlo method. By Central Limit Theorem\cite{book-statistical-inference},
$\frac{\sqrt{m_s}}{\sigma_i}(\operatorname{E}(\textbf{u}_i)-\textbf{u}_i)$ converges in probability to Gaussian distribution $G(0,1)$ with mean 0 and variance 1.
Then $\sum_{\textbf{u}_i\in \textbf{u}}[\frac{\sqrt{m_s}}{\sigma_i}(\textbf{u}_i-\operatorname{E}(\textbf{u}_i))]^2$ converges in probability to $\chi^2(|\textbf{u}|)$ where $\chi^2(|\textbf{u}|)$ is the
Chi-Square distribution with freedom $|\textbf{u}|$.
\end{proof}

\begin{corollary}
\label{corollary-error-MC}
The expected error of $\textbf{u}$ in Algorithm \ref{alg-MC} is
\begin{align*}
\operatorname{E}(error(\textbf{u}))<\frac{|\textbf{u}|}{m_s-1}max(\textbf{u})(1-max(\textbf{u}))
\end{align*}
\end{corollary}
\begin{proof}
Recall that each $\textbf{u}_i$ is the representation of
$Y_i=b(\textbf{Ay}+\frac{|\textbf{u}|}{2}-i < \theta \leq \textbf{Ay}+\frac{|\textbf{u}|}{2}-i+1)$. Thus the variance of $\textbf{u}_i$ should be calculated as
$\sigma_i^2=\frac{1}{m_s-1}\sum_{j=1}^{m_s}(Y_{ij}-\operatorname{E}(Y_i))^2$. Therefore,
\begin{multline*}
\sigma^2=max(\sigma_i^2)=max[\frac{1}{m_s-1}\sum_{j=1}^{m_s}(Y_{ij}-\textbf{u}_i)^2]\\
=max[\frac{1}{m_s-1}(\textbf{u}_im_s(1-\textbf{u}_i)^2+(1-\textbf{u}_i)m_s\textbf{u}_i^2)]\\
=max[\frac{m_s}{m_s-1}\textbf{u}_i(1-\textbf{u}_i)]\\
=\frac{m_s}{m_s-1}max(\textbf{u})(1-max(\textbf{u}))
\end{multline*}
Because the expectation of $\chi^2(|\textbf{u}|)$ is $\operatorname{E}(z)=|\textbf{u}|$, by Theorem \ref{theo-error-monte-carlo}, $\operatorname{E}(error(\textbf{u}))<\frac{|\textbf{u}|}{m_s-1}max(\textbf{u})(1-max(\textbf{u}))$.
\end{proof}

Because $\textbf{u}_i\in [0,\ 1]$, $max(\textbf{u})(1-max(\textbf{u}))\in [0,\ 0.25]$. The larger $|\textbf{u}|$, the smaller $m_s$, the bigger error of $\textbf{u}$, vice versa.
However, we cannot control $|\textbf{u}|$ and $max(\textbf{u})$, which are determined by $\textbf{A}$, $\textbf{S}$, $\boldsymbol\alpha$ and $m$. Thus the only way to
reduce the error of $\textbf{u}$ is to enlarge the sample size $m_s$.
But the computational complexity also rises with $m_s$, which is analyzed as follows.

\subsubsection{Complexity Analysis}
\begin{theorem}
\label{theo-complexity-1}
Algorithm \ref{alg-MC} takes time $O(m\cdot m_s)$.
\end{theorem}
\begin{proof}
To construct $m\times m_s$ random variables, it takes  $m\times m_s$ steps. The sum, round and hist operation also take
$m\times m_s$ each. So the computational time is $O(m\cdot m_s)$.
\end{proof}
Note that to guarantee a small error, a large $m_s$ is expected at the beginning.
However, the computational complexity is polynomial to $m$. The running time does not increase too fast when $m$ rises.

\subsection{Discrete Probability Calculation}
\label{sec-combination-PDF-2}

Because we know the noise $\tilde{\textbf{N}}_k$ has sensitivity $\textbf{S}_k$ and privacy budget $\boldsymbol\alpha_k$,
the cumulative distribution function
\begin{align*}
F_{\textbf{A}_{k}\tilde{\textbf{N}}_k}(z)=\frac{1}{2}exp(\frac{\boldsymbol\alpha_kz}{\textbf{S}_k\textbf{A}_{k}})
\end{align*}
Then we can calculate all the probability mass vectors for $\textbf{A}_{k}\tilde{\textbf{N}}_k$.

Next we define a function $conv(\textbf{u},\textbf{v})$ that convolves vectors $\textbf{u}$ and $\textbf{v}$. It returns a vector $\textbf{w}$ so that
\begin{align}\textbf{w}_k=\sum_j\textbf{u}_j\textbf{v}_{k-j+1}\\
\label{eqn-property-convolution}
|\textbf{w}|=|\textbf{u}|+|\textbf{v}|-1
\end{align}

Obviously the $conv()$ function corresponds to the convolution formula of probability function.
Therefore, $conv(\textbf{u},\textbf{v})$ returns the probability mass vector of $\textbf{u}+\textbf{v}$.

Recall the following equation, we can calculate the convolution result for all $\textbf{A}\tilde{\textbf{N}}$.
\begin{align}
\label{fml-TiN}
\textbf{A}\tilde{\textbf{N}}=\sum_{k=1}^{m} \textbf{A}_{k}\tilde{\textbf{N}}_k
\end{align}
Algorithm \ref{alg-combination-PDF} describes the process of discrete probability calculation.

\begin{algorithm}
\caption{Discrete probability calculation}
\begin{algorithmic}
\label{alg-combination-PDF}
\REQUIRE{$\textbf{A}$, $\tilde{\textbf{N}}$, $\textbf{S}$, $\boldsymbol\alpha$, $\gamma$}
\STATE
1. Construct all probability mass vectors of $\textbf{A}_{k}\tilde{\textbf{N}}_k$, $k=1,\cdots, m$:\\
\FOR{$k=1;k<m;k++$}
\STATE Choose length $|\textbf{v}|_k=ceil(\frac{|\textbf{A}_{k}|\textbf{S}_kln(m/\gamma)}{\boldsymbol\alpha_k})+1$
where function $ceil()$ rounds the input to the nearest integer greater than or equal to the input.\\
construct all the probability mass vectors $\textbf{v}_k$ for $\textbf{A}_{k}\tilde{\textbf{N}}_k$;\\
\ENDFOR\\
2. Convolve all $\textbf{v}$:
$\textbf{u}=conv(\textbf{v}_1,\cdots,\textbf{v}_{m})$;\\
\RETURN{probability mass vector $\textbf{u}$ as the discretized probability function of $\textbf{A}\tilde{\textbf{N}}$;}
\end{algorithmic}
\end{algorithm}

\begin{itemize}
\item
{\bf Example.}
Let $\gamma=0.01$. We can calculate the length of $\textbf{v}_1$ is $|\textbf{v}|_1=2*0.48*1*log(8/0.01)/0.05=128$. Then let $|\textbf{v}|_1=129$, meaning we only need to
construct the probability mass between $-64$ and $64$. We know the cumulative distribution function of Laplace distribution,
$Pr(z\leq\textbf{A}_{1}\tilde{\textbf{N}}_1<z+1)=\frac{1}{2}exp(\frac{\boldsymbol\alpha_1(z+1)}{\textbf{S}_1\textbf{A}_{1}})-\frac{1}{2}exp(\frac{\boldsymbol\alpha_1z}{\textbf{S}_1\textbf{A}_{1}})$.
Similarly, we can construct all the probability mass vectors for all $\textbf{A}_{k}\tilde{\textbf{N}}_k$. Finally we convolve them together.
The result $\textbf{u}$ represents the probability mass vector of $\theta$.
Figure \ref{Fig-probability-mass-vectors-PC} shows the the probability mass function of $\theta$.
\end{itemize}

\begin{figure}[h!]
\begin{minipage}{0.45\textwidth}
\centering
\includegraphics[width=8cm]{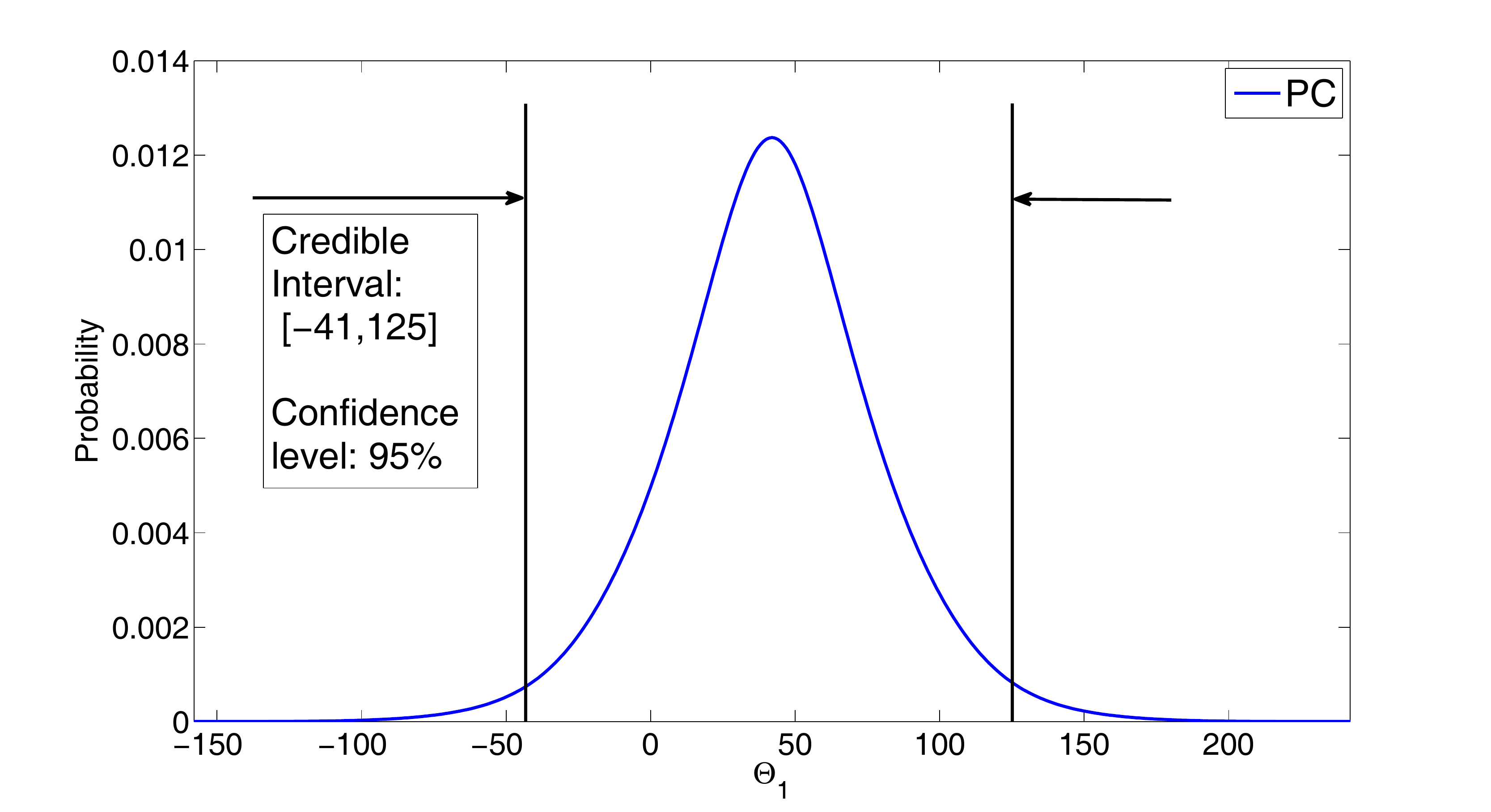}
\end{minipage}
\caption{Probability mass function of $\theta$.}
\label{Fig-probability-mass-vectors-PC}
\end{figure}

\subsubsection{Error Analysis}
\begin{lemma}
If probability loss $\mathcal{L}$ is defined to measure the lost probability mass of a probability mass vector $\textbf{v}_{\textbf{A}_{k}\tilde{\textbf{N}}_k}$, then
\begin{align}
\mathcal{L}(\textbf{v}_k)=exp(-\frac{\boldsymbol\alpha_k|\textbf{v}|_k}{2|\textbf{A}_{k}|\textbf{S}_k})=\gamma/m
\end{align}
\end{lemma}
\begin{proof}
\begin{align*}
\mathcal{L}(\textbf{v}_k)=\int_{-\infty}^{-|\textbf{v}|_k/2}\frac{\boldsymbol\alpha_k}{2|\textbf{A}_{k}|\textbf{S}_k}exp(\frac{\boldsymbol\alpha_kz}{|\textbf{A}_{k}|\textbf{S}_k})dz\\
+\int_{|\textbf{v}|_k/2}^{\infty}\frac{\boldsymbol\alpha_k}{2|\textbf{A}_{k}|\textbf{S}_k}exp(\frac{\boldsymbol\alpha_kz}{|\textbf{A}_{k}|\textbf{S}_k})dz\\
=exp(-\frac{\boldsymbol\alpha_k|\textbf{v}|_k}{2|\textbf{A}_{k}|\textbf{S}_k})
\end{align*}
\end{proof}

\begin{lemma}
The probability loss of Algorithm \ref{alg-combination-PDF} is
\begin{align}
\mathcal{L}(\textbf{u})=\sum_{k=1}^m exp(-\frac{\boldsymbol\alpha_k|\textbf{v}|_k}{2|\textbf{A}_{k}|\textbf{S}_k})=\gamma
\end{align}
\end{lemma}
Note $|\textbf{v}|_k$ is a predefined parameter, $|\textbf{v}|_k=-\frac{2\textbf{A}_{k}\textbf{S}_kln(\gamma/m)}{\boldsymbol\alpha_k}$,
but $|\textbf{u}|$ is the length of the returned $\textbf{u}$.

\begin{theorem}
\label{theo-error-linear}
The error of $\textbf{u}$ of in Algorithm \ref{alg-combination-PDF} satisfies
\begin{align*}
error(\textbf{u})< \gamma^2
\end{align*}
\end{theorem}
\begin{proof}
\begin{align*}
error(\textbf{u})=\sum_{\textbf{u}_i\in \textbf{u}} (\textbf{u}_i-\operatorname{E}(\textbf{u}_i))^2 < [\sum_{\textbf{u}_i\in \textbf{u}}(\textbf{u}_i-\operatorname{E}(\textbf{u}_i))]^2\\
=[\mathcal{L}(\textbf{u})]^2=\gamma^2
\end{align*}
\end{proof}
The trick to guarantee $error(\textbf{u})<\gamma^2$ is to choose $|\textbf{v}|_k$. However, the growth rate of $|\textbf{v}|_k$ with $m$ and $\gamma$ is very slow.
For example, even if $m=10^6$ and $\gamma=10^{-20}$, $-ln(\gamma/m)=60$. So it is important to conclude that
the error of Algorithm \ref{alg-combination-PDF} can be ignored.
The error in numerical calculation, which is inevitable for computer algorithm, would be even more noteworthy than $error(\textbf{u})$.
Although the error is not a problem,  the computational complexity increases fast.

\subsubsection{Complexity Analysis}
\begin{theorem}
\label{theo-complexity-2}
Let $\beta=sum^2(|\textbf{A}|diag(\textbf{S}./\boldsymbol\alpha))$,
Algorithm \ref{alg-combination-PDF} takes time
$O(log^2(m/\gamma)\beta)$.
\end{theorem}
\begin{proof}
To construct a probability mass vector, it needs $|\textbf{v}|_k$ steps where $|\textbf{v}|_k=-\frac{2|\textbf{A}_{k}|\textbf{S}_kln(\gamma/m)}{\boldsymbol\alpha_k}$.
Then
it takes $ln(m/\gamma)\sum_{k=1}^m{2|\textbf{A}_{k}|\textbf{S}_k}/\boldsymbol\alpha_k$ steps to construct all the probability mass vectors.
For any convolution $conv(\textbf{v}_i,\textbf{v}_j)$, without loss of generality, we assume $|\textbf{v}|_i\geq |\textbf{v}|_j$. Then it takes
$(1+|\textbf{v}|_j)|\textbf{v}|_j+(|\textbf{v}|_i-|\textbf{v}|_j)|\textbf{v}|_j=|\textbf{v}|_i|\textbf{v}|_j+|\textbf{v}|_j^2/2+|\textbf{v}|_j/2$.
Because we know the length of $|\textbf{v}|$ in all the convolution steps according to the property of convolution in equation (\ref{eqn-property-convolution}),
the first, $\cdots$ to the last convolution takes $|\textbf{v}|_2|\textbf{v}|_1+|\textbf{v}|_1^2/2+|\textbf{v}|_1/2$, $\cdots$, $(|\textbf{v}|_1+|\textbf{v}|_2+\cdots+|\textbf{v}|_{m-1})|\textbf{v}|_m+|\textbf{v}|_m^2/2+|\textbf{v}|_m/2$.
For the all convolutions, it takes $4(ln^2(m/\gamma)\sum_{j=1,k=1}^m\frac{|\textbf{A}_{j}\textbf{A}_{k}|\textbf{S}_j\textbf{S}_k}{\boldsymbol\alpha_j\boldsymbol\alpha_k})$.
$\sum_{j=1,k=1}^m\frac{|\textbf{A}_{j}\textbf{A}_{k}|\textbf{S}_j\textbf{S}_k}{\boldsymbol\alpha_j\boldsymbol\alpha_k}$ can be written as
$\frac{1}{2}(\sum_{k=1}^m|\textbf{A}_{j}\textbf{S}_j/\boldsymbol\alpha_j|)^2-\frac{1}{2}\sum_{k=1}^m|\textbf{A}_{j}\textbf{S}_j/\boldsymbol\alpha_j|^2$, then the complexity can be
written as $O(log^2(m/\gamma)sum^2(|\textbf{A}|diag(\textbf{S}./\boldsymbol\alpha)))$.
Because by Theorem \ref{theo-hat-theta-BLUE} $\textbf{A}$ is determined by $\textbf{H}$, $\textbf{Q}$, $\boldsymbol\alpha$ and $\textbf{S}$, the computational
complexity is also determined by $\textbf{H}$, $\textbf{Q}_i$, $\boldsymbol\alpha$ and $\textbf{S}$.
\end{proof}

An important conclusion which can be drawn from the above analysis is to convolve all $\textbf{v}$ from the shortest to the longest. Because the running time
is subject to the length of vectors, this convolution sequence takes the least time.

\subsection{Summary}
To acquire a small error $\gamma^2$, the big sample size of Monte Carlo method would always be expected.
But the error requirement can be easily satisfied by probability calculation method without significant cost\footnote{We ignore the error of numerical calculation error in computer algorithms}.
However, the running time of Monte Carlo method increases slower than probability calculation method.
Therefore, when $log^2(m/\gamma)\beta$ is small,
probability calculation method would be better. When $log^2(m/\gamma)\beta$ becomes larger, Monte Carlo method would be more preferable.

\section{Application}
The proposed inference mechanism can be used in many scenarios such as parameter estimation, hypothesis testing and statistical inference control.
For example, to control the inference result of an adversary who assumes a claim $\mathcal{C}$ that the $\theta \in [L_0, U_0]$,
one may want to obtain the probability $Pr(\mathcal{C})$ using the inference method.
Once $Pr(\mathcal{C})$ becomes very high, it means the attacker is confidence about the claim $\mathcal{C}$ to be true. Then further queries
should be declined or at least be carefully answered(only returning the inference result of a new-coming query is one option).
In this way, the
statistical inference result can be quantified and controlled.
In this paper, we show how to apply the mechanism into a utility-driven query answering system.
Following assumption(which holds for all differentially private system) are made at the first place.

We assume the cells of the data are independent or only have negative correlations
\cite{Rastogi-adversarial-privacy}.
Generally speaking, this means adversaries cannot infer the answer of a cell from any other cells. specifically, cells are independent if they
do not have any correlation at all; cells are negative correlated\footnote{For two records $\textbf{v}_1$ and $\textbf{v}_2$, it's negative correlated if $Pr(\textbf{v}_1=t)\geq Pr(\textbf{v}_1=t|\textbf{v}_2=t$).}
 if the inference result is within the previous knowledge of the adversary. So the
inference from correlation does not help.
In this case,
the composability of privacy cost\cite{McSherry-PINQ} holds as follows.

\begin{theorem}[Sequential Composition \cite{McSherry-PINQ}]
\label{lemma-composition}
Let $M_i$ each provide $\boldsymbol\alpha_i$-differential privacy.
The sequence of $M_i$ provides $(\sum_i{\boldsymbol\alpha_i})$-differential privacy.
\end{theorem}

\begin{theorem}[Parallel Composition \cite{McSherry-PINQ}]
\label{lemma-parallel}
If $D_i$ are disjoint subsets of the original database and $M_i$ provides $\boldsymbol\alpha_i$-differential privacy for each $D_i$,
then the sequence of $M_i$ provides $max(\boldsymbol\alpha)$-differential privacy.
\end{theorem}
%
%
%

\subsection{Framework of Query-answering System}
We allow a user require the utility of an issued query $Q$ in the form of $1-\delta$ credible interval with length less than $2\epsilon$.
If the derived answer satisfies user's requirement, then no budget needs to be allocated because the estimation can be returned.
Only when the estimation can not meet the utility requirement, the query mechanism is invoked for a differentially private answer.
In this way, not only the utility is guaranteed for users, but also the privacy budget can be saved so that the lifetime of the system can be extended.

We embed a budget allocation method in the query mechanism to calculate the minimum necessary budget for the query.
Thus the privacy budget can be conserved.
Meanwhile, to make sure the query mechanism can be used to answer the query,
the overall privacy cost of the system is also.
Only when it does not violate the privacy bound to answer the query, the query mechanism can be invoked.

The framework of our model is illustrated in Figure \ref{Fig_Frame_online}.

\begin{figure}[h!]
\hspace{-0.8cm}
\includegraphics[width=10cm]{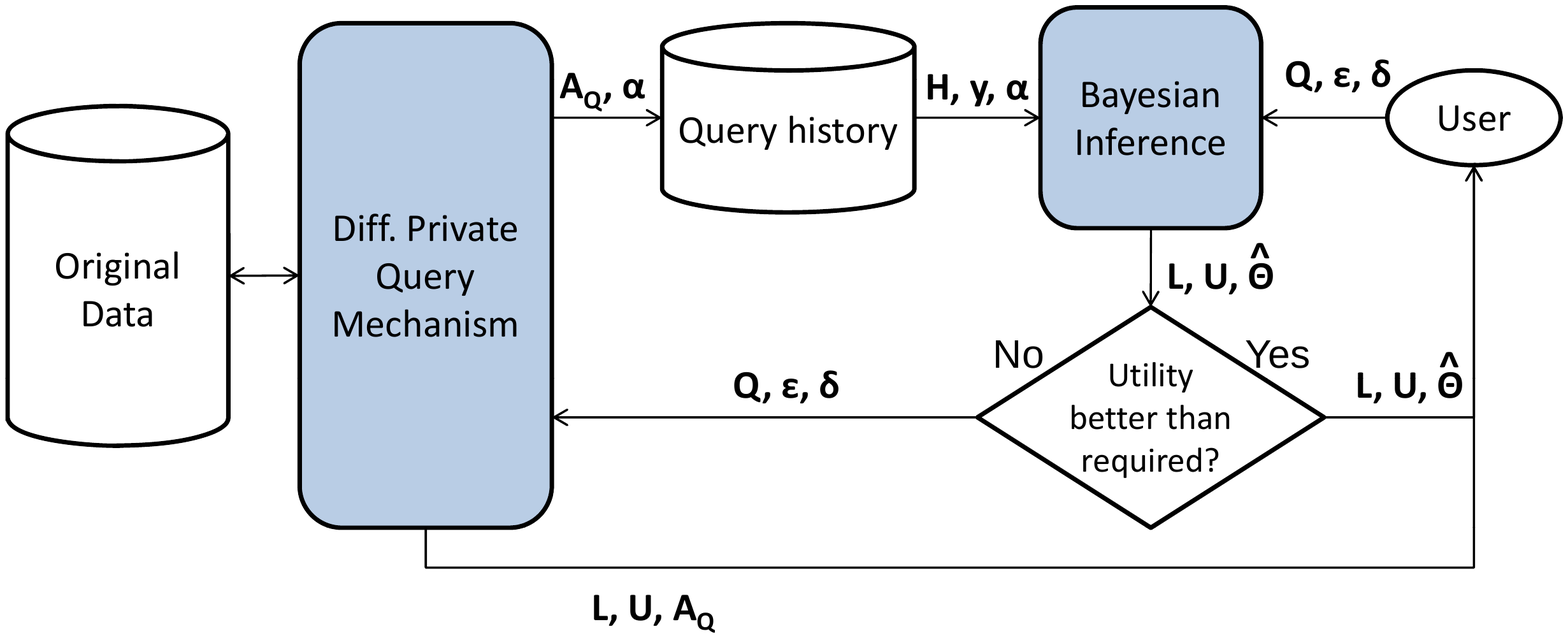}
\vspace{-3.6cm}
\caption{Framework of query-answering system}
\label{Fig_Frame_online}
\end{figure}

\subsection{Confidence Level and Credible Interval}
\label{cha3:sec-error}

If we have a certain interval $[L,\ U]$, then the confidence level can be derived by the probability function of Bayesian inference.
The approach is to cumulate all the probability mass in this interval. Then the summation of probability mass would be the confidence level.

In the query-answering system, we let users demand $1-\delta$ credible interval.
This problem can be formulated as following:
given the probability function and confidence level $1-\delta$, how to find the narrowest interval $[L,\ U]$ containing probability $1-\delta$?
First, it's easy to prove that the probability mass vector is symmetric, which means the probability function of $\theta$ would always be an even function.
Because each Laplace noise $\tilde{\textbf{N}}_i$ has symmetric probability distribution, the summation $\textbf{A}\tilde{\textbf{N}}$ is also symmetric.
The midpoint of probability mass vector always has the highest probability. Then the approach is to cumulate the probability mass from the midpoint to the left or right side until it reaches $\frac{1-\delta}{2}$.
Algorithm \ref{alg-usefulness} summarizes the process.

\begin{algorithm}
\caption{Find the credible interval}
\begin{algorithmic}
\label{alg-usefulness}
\REQUIRE{probability mass vector $v$, $\delta$, $\textbf{A}$, $\textbf{y}$, i}\\
\STATE 1. $L\leftarrow \textbf{Ay}$; $U \leftarrow \textbf{Ay}$; $C_m\leftarrow \textbf{u}[|\textbf{u}|/2+\frac{1}{2}]$;\\
2. \WHILE{$C_m<1-\delta$}
\STATE
$L\leftarrow L-1$; $U\leftarrow U+1$\\
$C_m\leftarrow C_m+2\textbf{u}[|\textbf{u}|/2+\frac{1}{2}+L-\textbf{Ay}]$;\\
\ENDWHILE
\RETURN{$L,\ U$}
\end{algorithmic}
\end{algorithm}
\begin{itemize}
\item
{\bf Example.}
Assume a user is interested in two questions: 1. what is the $95\%$ credible interval of $\theta$;
2. What is the probability that $\theta>0$?

For question 1, we use Algorithm \ref{alg-usefulness} to derive the credible interval as $[-41,\ 125]$. For question 2,
we calculate $sum(\textbf{u}(|\textbf{u}|/2-\frac{1}{2}-round(\theta)+0),\cdots, \textbf{u}(|\textbf{u}|))$, so $Pr(\theta>0)=0.88$.
Figure \ref{Fig-probability-mass-vectors-PC} shows the $95\%$ credible interval.
\end{itemize}

\subsection{Budget Allocation Method}
Only when the utility of estimate does not meet the requirement, we need to invoke the differentially private query mechanism,
return $\mathcal{A}(D)$ and interval $[\mathcal{A}(D)-\epsilon, \mathcal{A}(D)+\epsilon]$ to the user,
add the query $Q$ to query history $\textbf{H}$, noisy answer $\mathcal{A}(D)$ to $\textbf{y}$ and allocated budget to $\boldsymbol\alpha$.

Theorem \ref{theo-budget-allocation} shows the budget allocation method which spends
$S_Q \frac{-ln\delta}{\epsilon}$ privacy cost.
\begin{theorem}
\label{theo-budget-allocation}
Given the utility requirement $1-\delta$ credible interval with length less than $2\epsilon$,
it's enough to allocate $\alpha=S_Q \frac{-ln\delta}{\epsilon}$.
\end{theorem}
\begin{proof}
Let $\mathcal{A}_Q(D)=Q(D)+\tilde{N}(\alpha/S_Q)$.
Because probability function of laplace distribution is symmetric, in definition \ref{definition-usefulness}, $[\mathcal{A}(D)-\epsilon, \mathcal{A}(D)+\epsilon]$ gives the most probability mass.
\begin{multline}
\label{fml-budget-allocation-1}
Pr(-\epsilon\leq \mathcal{A}_Q(D)-Q(D)\leq \epsilon)\\
=Pr(-\epsilon\leq \tilde{N} \leq \epsilon)\\
=\int_{-\epsilon}^{\epsilon}\frac{\alpha/S_Q}{2}exp(-{\alpha/S_Q |x|})dx\\
=2\int_{-\epsilon}^0 \frac{\alpha/S_Q}{2}exp({\alpha/S_Q x})dx\\
=1-exp(-\frac{\epsilon \alpha/S_Q}{2})\geq 1-\delta \quad
\end{multline}
So we have $\alpha\geq S_Q \frac{-ln\delta}{\epsilon}$.
\end{proof}

\subsection{Privacy Cost Evaluation}

Note that the system should also be a bounded differentially private system. Therefore, we cannot answer unlimited queries.
Given a privacy budget as the total bound, we should also compute the current privacy cost of the system to test whether the privacy cost
is within privacy budget if answering the query $Q$.

The system privacy cost can be derived by Theorem \ref{theo-privacy-cost}.
\begin{theorem}
\label{theo-privacy-cost} Given current differentially private
system with query history , denoted by $\textbf{H}$ and privacy cost
of all query , denoted by $\boldsymbol\alpha$, the system privacy cost,
denoted by $\bar{\alpha}$, can be derived from equation
(\ref{fml-system-privacy}, \ref{fml-B}).
\begin{align}
\label{fml-system-privacy} \bar{\alpha}=max(\textbf{B})
\\
\label{fml-B}
\textbf{B}=(\boldsymbol\alpha./\textbf{S})^Tabs(\textbf{H})
\end{align}
where $\textbf{B}$ be vector of used privacy budget of each cell,
function $diag()$ transforms a vector to a matrix with each of
element in diagonal, $Sum_{row}$ means the sum of rows.
\end{theorem}

\begin{proof}
First, we prove for each query $H_i$ in $\textbf{H}$, the privacy
budget cost of each cell $x_j$ is $\boldsymbol\alpha_i/\textbf{S}_i*abs(H_i)$; then
add up each cell and each row, we get the result.

For any query $Q$, the differentially private answer is returned as
$\mathcal{A}_Q$. Assume an adversary knows all the records of the
data set except one record in cell $x_j$, we define the local
privacy cost for cell $x_j$ is $\boldsymbol\alpha_j$
 if $\mathcal{A}_Q$ is $\boldsymbol\alpha_j$-differentially private where
\begin{align}
\label{fml-budget-cost} exp(\boldsymbol\alpha_j)=\operatorname{sup}_{r\subseteq
Range(Q)}\frac{Pr[\mathcal{A}_Q(D)=r | D=D_1]}{Pr[\mathcal{A}_Q(D)=r
| D=D_2]}
\end{align}
where $D_1$ and $D_2$ are neighboring databases the same as
definition \ref{def-DP}. Next we compute the local privacy cost of
each cell for each query.

For each query $\textbf{H}_i$, the Laplace noise is
$\tilde{N}(\boldsymbol\alpha_i/S_{H_i})$ according to equation
(\ref{fml-lap-mec}). By lemma \ref{lemma-GS} and equation
(\ref{fml-SH}), it is $\tilde{N}(\boldsymbol\alpha_i/\textbf{S}_i)$. For
the cell $x_j$, assume an adversary knows all but one record in
$x_j$, then the count numbers of other cells become constant in this
query. Thus this query becomes
$\textbf{y}_i=\textbf{H}_{ij}\textbf{x}_j+\tilde{N}(\boldsymbol\alpha_i/\textbf{S}_i)+constant$.
The $S$ of this query is $H_{ij}$ and the noise of this query is
$\tilde{N}(\boldsymbol\alpha_i/\textbf{S}_i)$. By Theorem
\ref{theo-Lap-DP}, it guarantees $H_{ij}
\boldsymbol\alpha_i/\textbf{S}_i$-differential privacy. Then the budget cost of
$x_j$ is $H_{ij} \boldsymbol\alpha_i/\textbf{S}_i$.

By theorem \ref{lemma-composition}, the local privacy cost composes
linearly for all queries. Then the local privacy costs can add up by
each query.

Because we assume cells are independent or negatively correlated in
this paper,
 by Theorem \ref{lemma-parallel}, the privacy costs of all cells do not affect each other.
Therefore, we prove equation (\ref{fml-B}).

The highest privacy cost indicates the current system privacy cost,
denoted by $\bar{\alpha}$. So we prove equation
(\ref{fml-system-privacy}).
\end{proof}

\begin{itemize}
\item
{\bf Example.}
In equation (\ref{fml-H}), $\textbf{H}_6=[2\ 1\ 0\ 0]$, $\textbf{x}_1=10$, $\textbf{x}_2=20$ by equation
(\ref{eqn_x}), $\boldsymbol\alpha_6=0.05$ by equation (\ref{fml-A}).
According to equation (\ref{fml-y}), $\textbf{y}_6=\textbf{H}_6
\textbf{x}+\tilde{N}(\boldsymbol\alpha_6/2)=2\textbf{x}_1+1\textbf{x}_2+\tilde{N}(0.05/2)$.
For the cell $\textbf{x}_1$ with coefficient 2, assume an adversary
knows all but one record in $\textbf{x}_1$, then the count numbers
of other cells become constant in this query. So
$\textbf{y}_3=20+2\textbf{x}_1+\tilde{N}(0.05/2)=2\textbf{x}_1+\tilde{N}(0.05/2)$.
The $S$ of this query is 2 and the noise of this query is
$\tilde{N}(0.05/2)$. According to Theorem \ref{theo-Lap-DP}, it
guarantees $0.05$-differential privacy. Then the local privacy cost
of $\textbf{x}_1$ is $0.05$. Similarly for $\textbf{x}_2$ with
coefficient 1, $\textbf{y}_3=20+1\textbf{x}_2+\tilde{N}(0.05/2)$, so
the local privacy cost of $\textbf{x}_2$ is $0.025$; for
$\textbf{x}_3$ with coefficient 0,
$\textbf{y}_3=40+0*\textbf{x}_3+\tilde{N}(0.05/4)$, so the local
privacy cost of $\textbf{x}_3$ is 0.

We obtain $\textbf{B}=[0.1, 0.275,     0.25,     0.375]$, thus $\bar{\alpha}=0.375$. 
\end{itemize}


\section{Experimental study}
\label{sec-exp}
\subsection{Experiment Setup}
\textbf{Environment}
All the algorithms are written in MATLAB. All the functions in this paper are consistent with MATLAB functions.
All experiments were run on a
computer with Intel P8600(2 * 2.4 GHz) CPU and 2GB memory.

\begin{table}
\label{tbl-denotation}
\centering
\begin{tabular}{cc}
\hline
\multicolumn{2}{c}{{\bf Denotation}}\\
\hline
m& $\#$ of records in history, $\textbf{H}\in \mathbb{R}^{m\times n}$\\
n& $\#$ of cells, $\textbf{H}\in \mathbb{R}^{m\times n}$\\
$m_s$ & sample size for Monte Carlo method\\
$\beta$ & $\beta=sum^2(|\textbf{A}|diag(\textbf{S}./\boldsymbol\alpha))$\\
$\bar{\alpha}$ & system privacy cost\\
$R_a$&answering ratio defined in Definition \ref{def-answer-ratio}\\
$R_i$&ratio of satisfied answers in Definition \ref{def-ratio-of-satisfaction}\\
$E$& relative error in Definition \ref{def-relative-error}\\
\hline
\end{tabular}
\end{table}

\textbf{Denotation} Denotations are summarized in Table \ref{tbl-denotation}. The series a:b:c means the series from a to c with space b.
For example, 100:100:500 means 100, 200, 300, 400, 500.

\textbf{Settings} The efficiency of  BLUE is shown in section \ref{sec-exp-point}. Next the accuracy and efficiency of Bayesian inference
is shown in section \ref{sec-exp-interval}. Finally, we measure the overall performance of a utility-driven online query answering system in section \ref{sec-exp-system}.

\subsection{BLUE}
\label{sec-exp-point}
Because the linear solution $\hat{\theta}=\textbf{Ay}$ is used in following steps, running time of BLUE is measured.
\subsubsection{Running Time vs. $m$}
Let
$m=100:100:2000$, $n=100$, $\boldsymbol\alpha$ be a vector of random variables uniformly distributed in $[0,\ 1]$, $max(\textbf{S})=10$.
The impact of $m$ on the running time is shown in Figure \ref{Fig_MSE_BLUE_time}(left). Running time increases slowly with $m$.
\begin{figure}[h!]
\begin{minipage}{0.23\textwidth}
\centering
\includegraphics[width=4cm]{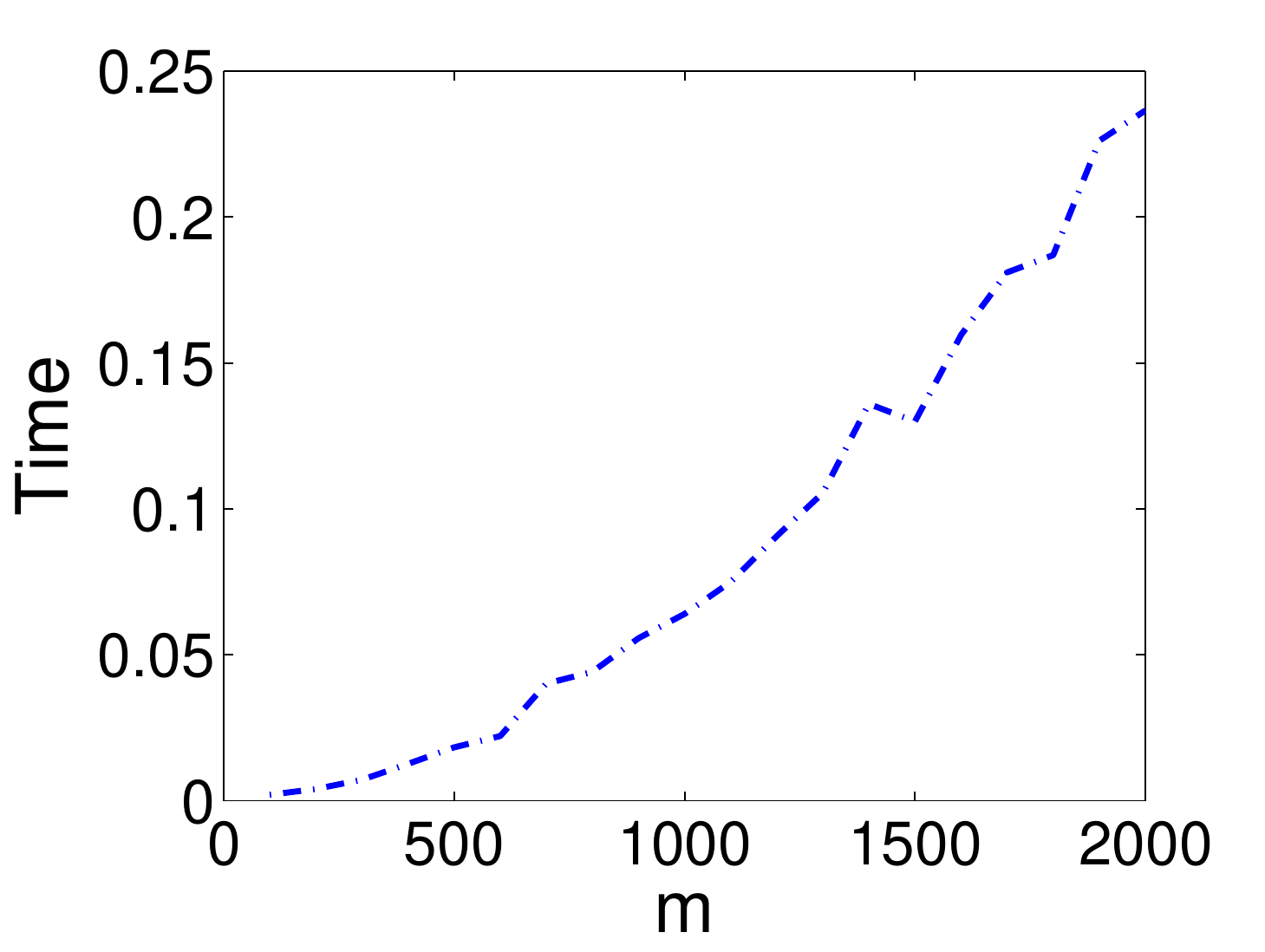}
\end{minipage}
\begin{minipage}{0.23\textwidth}
\includegraphics[width=4cm]{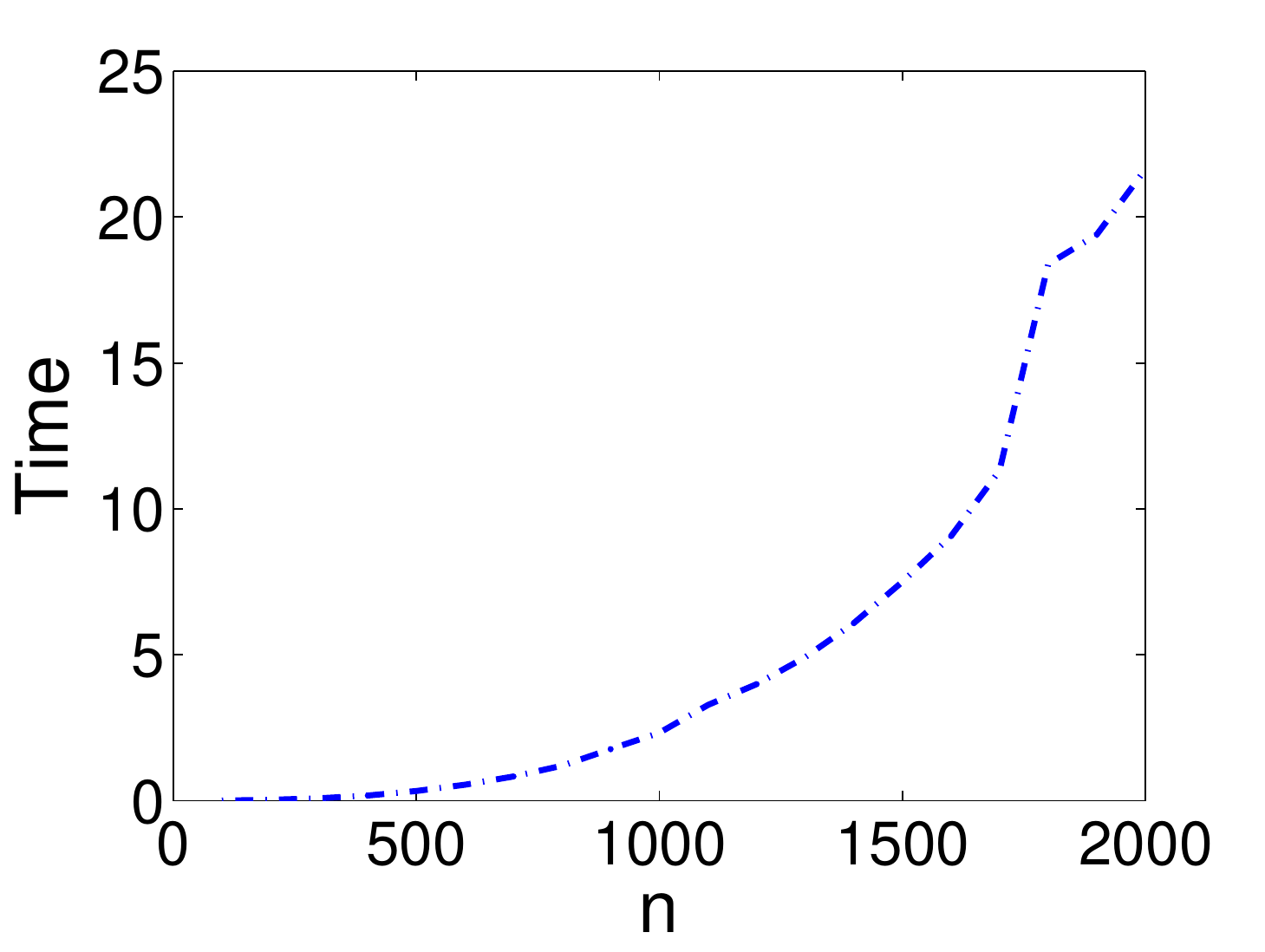}
\end{minipage}
\caption{Running time vs. $m$(left) and $n$(right)}
\label{Fig_MSE_BLUE_time}
\end{figure}
\subsubsection{Running Time vs. $n$}
Let
$n=100:100:2000$, $\boldsymbol\alpha$ be a vector of random variables uniformly distributed in $[0,\ 1]$, $max(\textbf{S})=10$.
Because the linear system should be overdetermined, which means $m>n$, we set  $m=2n$. From Figure \ref{Fig_MSE_BLUE_time}(right) we know that
running time rise significantly with $n$, which actually has verified the motivation of \cite{Improving-utility-PCA} that proposed the PCA(Principle Component Analysis) and maximum entropy methods
to boost the speed of this step.

\subsection{Bayesian Inference}
\label{sec-exp-interval}
By Corollary \ref{corollary-error-MC}, we can derive a bound for $m_s$,
 $m_s>\frac{|u|}{\gamma^2}max(u)(1-max(u))+1$. However, this is not a sufficient bound because to implement Monte Carlo method a large number of sample size is needed in the first place.
Therefore, we set the minimal sample size to be $10^4$. When the derived bound is larger than $10^4$, let $m_s$ be it. So $m_s=10^4<(\frac{|u|}{\gamma^2}max(u)(1-max(u))+1)?(\frac{|u|}{\gamma^2}max(u)(1-max(u))+1):10^4$.
\subsubsection{Running Time vs. $m$}
Let
$\gamma=0.01$,  $max(S)=5$, $n=100$, $m=200:100:1000$. Figure \ref{Fig_time_MC_PC_m}(left) shows that running time of MC rises with $m$, which complies with Theorem \ref{theo-complexity-1}.
But $m$ does not impact the running time of PC significantly.
\begin{figure}[h!]
\begin{minipage}{0.23\textwidth}
\centering
\includegraphics[width=4cm]{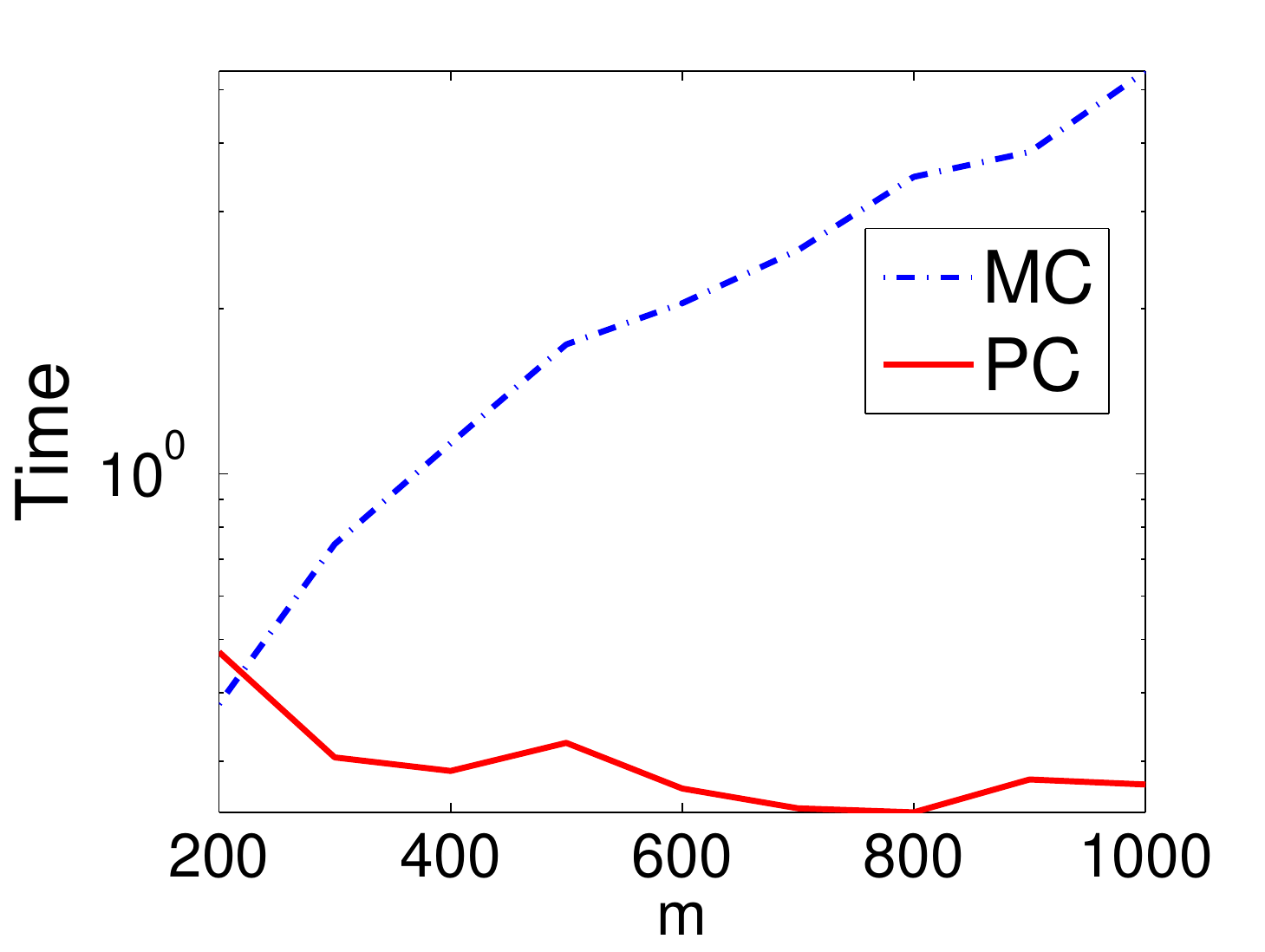}
\end{minipage}
\begin{minipage}{0.23\textwidth}
\includegraphics[width=4cm]{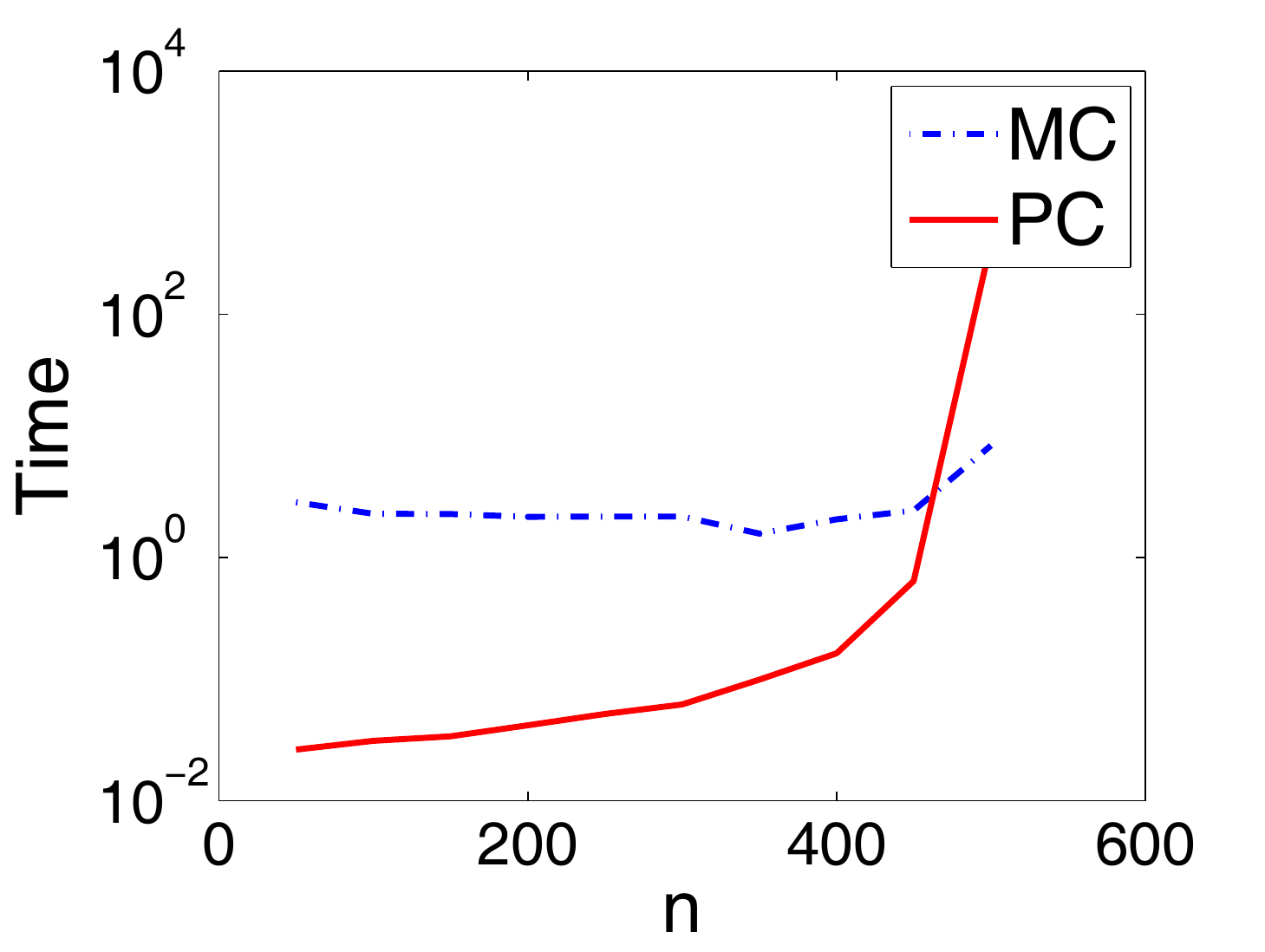}
\end{minipage}
\caption{Running time vs. $m$(left) and $n$(right)}
\label{Fig_time_MC_PC_m}
\end{figure}
\subsubsection{Running Time vs. $n$}
Let $\gamma=0.01$, $max(S)=5$, $m=500$, $n=50:50:100$. Impact of $n$ is shown in Figure \ref{Fig_time_MC_PC_m}(right). Running time PC increases fast with $n$.
In contrast, parameter $n$ does not affect running time of MC.

\subsubsection{Time and Error of Method PC}
Let $\gamma=0.01$, $max(\textbf{S})=5$, $m=1000$, $n=100$, $\beta=sum^2(|\textbf{A}_i|diag(\textbf{S}./\boldsymbol\alpha))$.
Theorem \ref{theo-complexity-2} implies that running time of PC is affected by $\beta$.
Experimental evaluation in Figure \ref{Fig_time_MC_PC_S_alpha}(left) confirms the result.
The error defined in Definition \ref{def-error-u} is also measured in Figure \ref{Fig_time_MC_PC_S_alpha}(right).
Because we can set $\gamma$ in the first place, the error of PC does not increase with $\beta$ and any other parameters.
\begin{figure}[h!]
\begin{minipage}{0.23\textwidth}
\centering
\includegraphics[width=4cm]{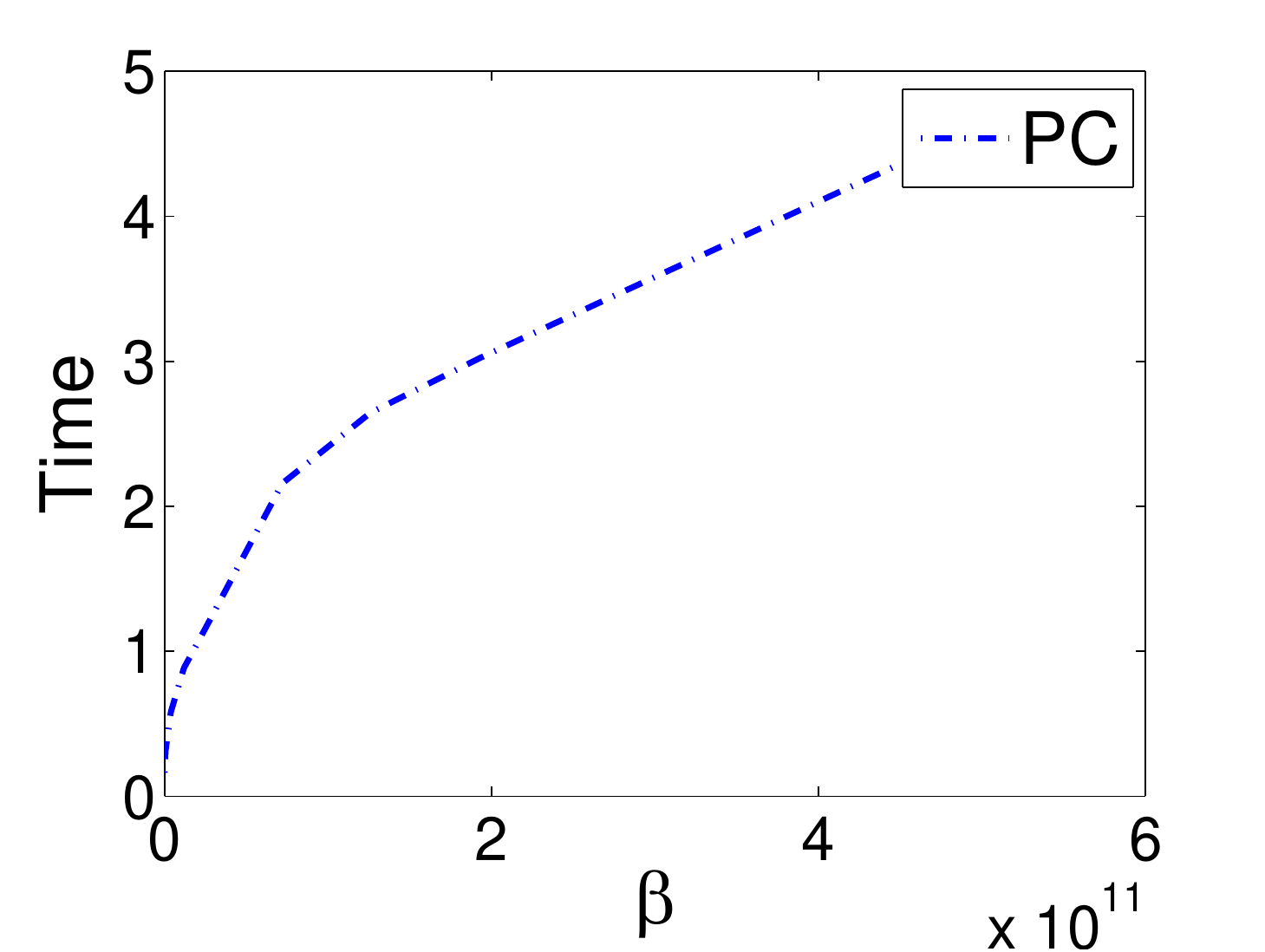}
\end{minipage}
\begin{minipage}{0.23\textwidth}
\centering
\includegraphics[width=4cm]{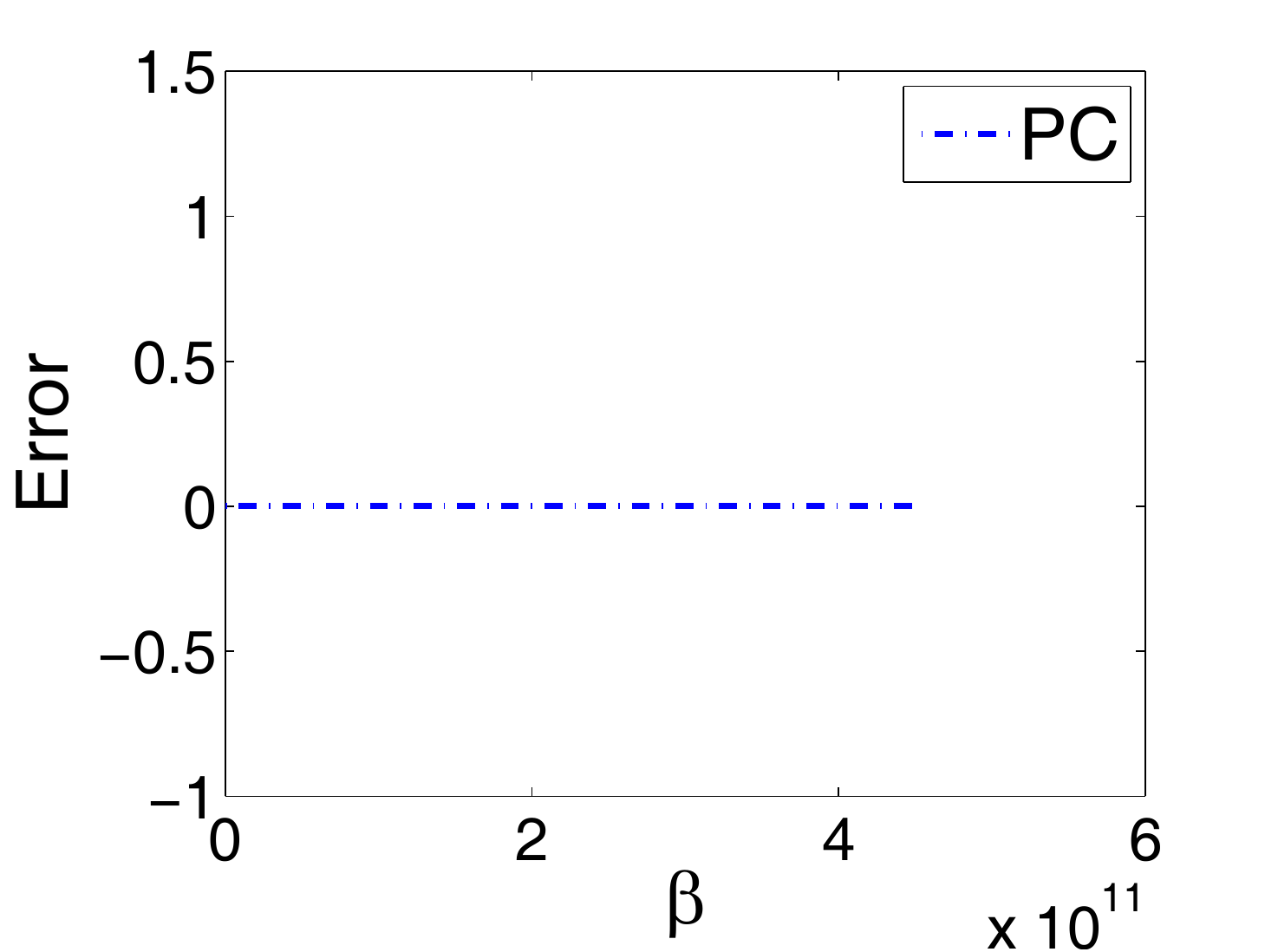}
\end{minipage}
\caption{Time(left) and Error(right) of PC vs. $\beta$}
\label{Fig_time_MC_PC_S_alpha}
\end{figure}
\subsubsection{Time and Error of Method MC}
Let $n=100$. The sample size $m_s$ impacts both time and error of MC method. Let $m_s=10^5:10^5:10^6$. The time and error
is evaluated in Figure \ref{Fig_time_MC_sample_size} on left and right respectively. When $m_s$ increases,
the running time rises and the error declines, which complies with Theorem \ref{theo-complexity-1} and Corollary\ref{corollary-error-MC}.
\begin{figure}[h!]
\begin{minipage}{0.23\textwidth}
\centering
\includegraphics[width=4cm]{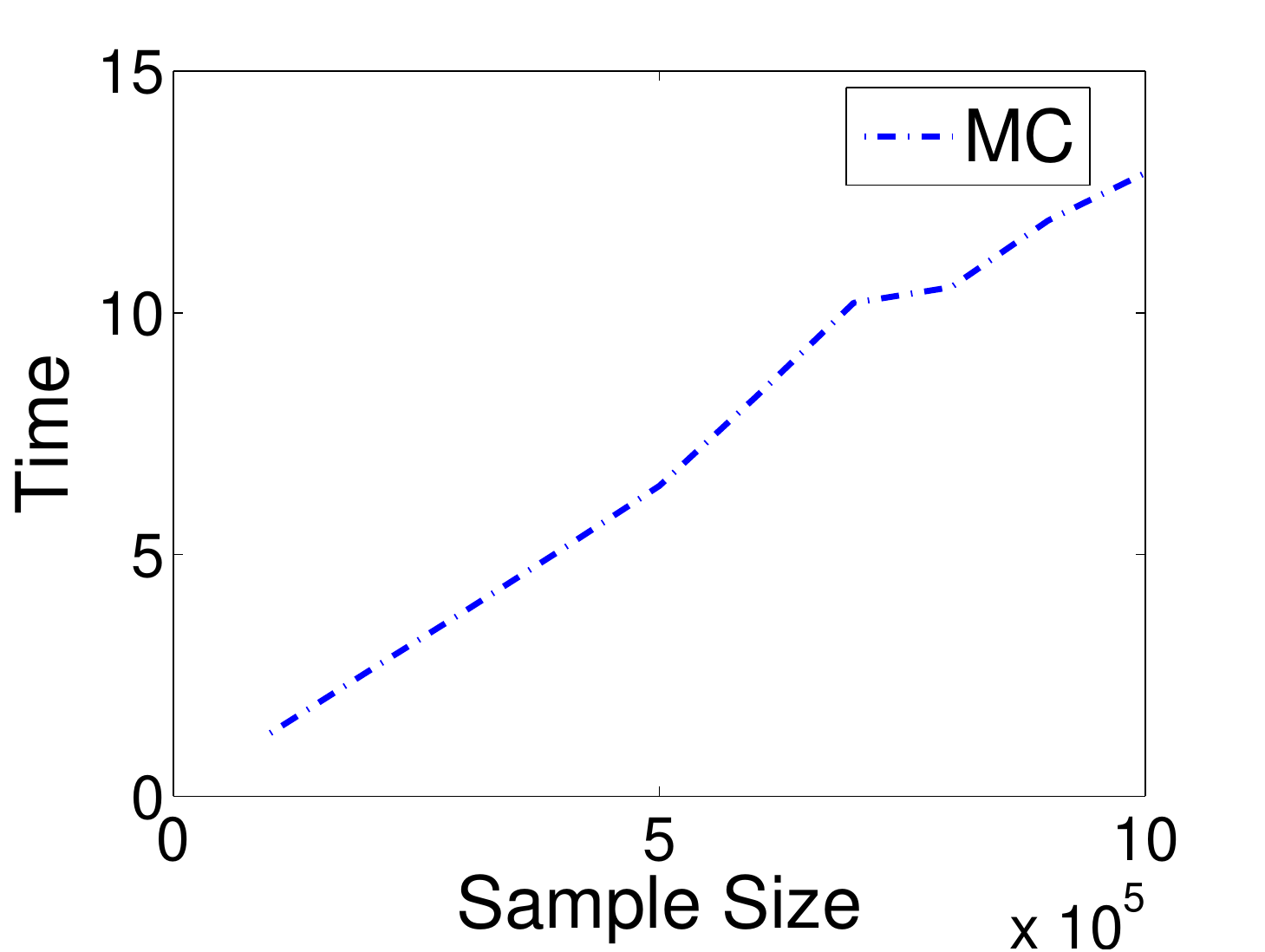}
\end{minipage}
\begin{minipage}{0.23\textwidth}
\centering
\includegraphics[width=4cm]{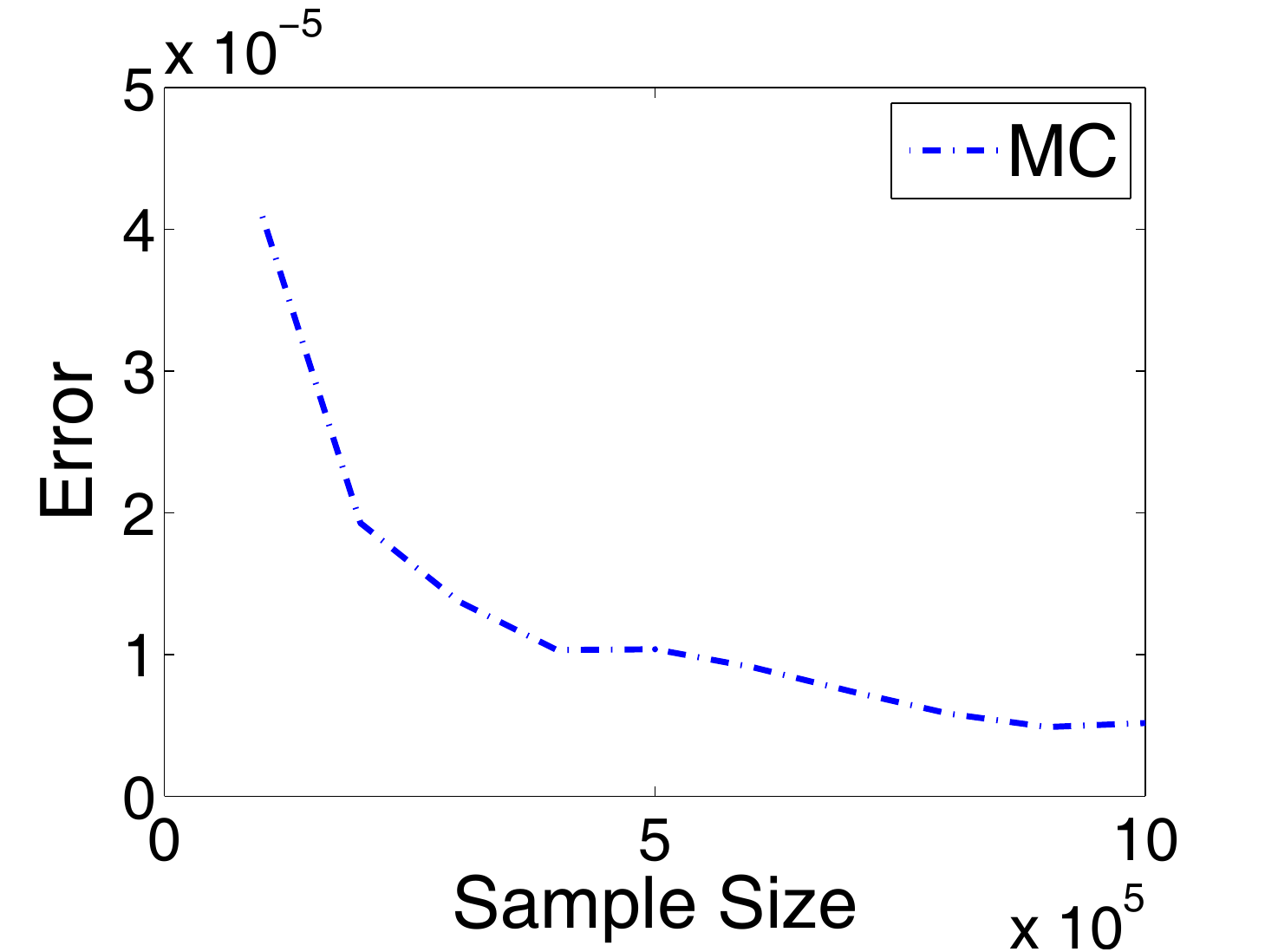}
\end{minipage}
\caption{Time(left) and Error(right) of MC vs. sample Size}
\label{Fig_time_MC_sample_size}
\end{figure}

\subsection{Utility-driven Query Answering System}
\label{sec-exp-system}
\textbf{Data.} We use the three data sets: net-trace, social network and search logs, which are the same data sets with \cite{boost-accuracy}.
Net-trace is an IP-level network trace collected at
a major university; Social network is a graph of
friendship relations in an online social network site; Search
logs is a set of search query logs over time from Jan.1, 2004 to 2011.
The results of different data are similar, so we only show net-trace and search logs results for lack of space.

\textbf{Query.} User queries may not necessarily obey a certain distribution. So the question arises how to generate the testing queries.
In our setting, queries are assumed to be multinomial distributed,
which means the probability for the query to involve cell $x_j$ is $P_j$, where $\sum_{j=1}^n P_j=1$.
As shown in Figure \ref{Fig_MSE_BLUE_time}(right), the running time rises dramatically with $n$.
For efficiency and practicality, we reduce the dimensionality by assuming some regions of interest, where the cells are asked more frequently than others.
Sparse distributed queries can also be estimated with the linear solution in \cite{Improving-utility-PCA}.
Equation (\ref{fml-Query-P}) shows the probability of each cell.
\begin{align}
\label{fml-Query-P}
P_j=0.9 * 10^{-floor(\frac{j-1}{10})}
\end{align}
where the function $floor()$ rounds the input to the nearest integers less than it.

To generate a query, we first generate a random number $n_t$ in $1\sim 10$, which is the $\#$ of independent trails of multinomial distribution.
Then we generate the query from the multinomial distribution, equation (\ref{fml-Q-Pr}) shows the probability of $Q$.
\begin{multline}
\label{fml-Q-Pr}
Pr(Q_1=q_1, Q_2=q_2,\cdots, Q_n=q_n)=\\\frac{n_t!}{q_1! q_2! \cdots q_n!} P_1^{q_1} P_2^{q_2}  \cdots P_n^{q_n}
\end{multline}
where $\sum_j q_j=n_t$.

Finally, we generate 1000 queries for each setting.

\textbf{Utility Requirement}
We assume each query has different utility requirement ($\epsilon$,$\delta$).
Note that $\epsilon$ and $\delta$ determines $\boldsymbol\alpha$ and $\boldsymbol\alpha$ determines $\beta$. Therefore the running time of
PC method is related to $\epsilon$ and $\delta$. Thus we assume $\epsilon$ has a upper bound $10^3$ and is uniformly distributed.

\textbf{Metrics}
Besides system privacy cost $\bar{\alpha}$
in theorem \ref{theo-privacy-cost}, following metrics are used.

For a bounded system, Once the system privacy cost $\bar{\alpha}$ reaches the overall privacy budget, then the system cannot answer further queries.
To measure this, we define the ratio of answered queries as below. In another word, it indicates the capacity(or life span) of a system.
\begin{definition}
\label{def-answer-ratio}
Under the bound of overall privacy budget, given a set of queries \textbf{Q} with utility requirements, the answering ratio $R_a$ is
\begin{align}
R_a=\frac{(\#\ of\ answers:\ Pr(L\leq \theta \leq U)\geq 1-\delta)}{(\#\ of\ all\ queries)}
\end{align}
\end{definition}

Among all the {\bf answered queries}, we want to know the accuracy of the returned answer.
We use following two metrics.
1. $R_i$ is defined to show whether the returned interval really contain the original answer;
2. $E$ is the distance between returned answer and true answer.
It is easy to prove that if $\#$ of answered queries approaches to infinity, $R_i$ converges to confidence level.

\begin{definition}
\label{def-ratio-of-satisfaction}
For a set of queries with utility requirements, the credible intervals $[L,U]$ are returned.
Ratio of reliability $R_i$ is defined as follows:
\begin{align}
R_i=\frac{\sum_{i \in answered\ queries}(\#\ of\ answers:\ L\leq \theta \leq L+2\epsilon_i)}{\#\ of\ all\ answered\ queries}
\end{align}
\end{definition}
Note that in the experiment we know the true value of $\theta$. Therefore, $R_i$ can be derived.

\begin{definition}
\label{def-relative-error}
For each query with ($\epsilon,\delta$) requirement, a returned answer $\hat{\theta}$ is provided by the query answering system.
If the original answer is $\theta$,
the relative error is defined to reflect the accuracy of $\hat{\theta}$:
\begin{align}
\label{eqn-relative-error}
E=\frac{\sum_{i \in answered\ queries}(|\hat{\theta}_i-\theta_i|/(2 \epsilon_i))}{\#\ of\ all\ answered\ queries}
\end{align}
\end{definition}

In summary, $\bar{\alpha}$ indicates the system privacy cost; $R_a$ indicates the capacity(or life span) of a system;
$R_i$ indicates the confidence level of returned credible interval $[L,U]$; $E$ indicates the accuracy of the point estimation $\hat{\theta}$.

\textbf{Settings}
We use PC inference method for efficiency.
The above four metrics are evaluated in two settings. In the first setting, the overall privacy budget is unbounded.
Thus $R_a=1$ because all queries can be answered. Also a set of history queries,
known as $\textbf{H}$, is also given in advance. Although it is not a realistic setting, we can measure how the system can save privacy budget
by our method. In the second setting, the privacy budget is bounded. Thus $R_a\leq 1$.

\subsubsection{Unbounded Setting}
A hierarchical partitioning tree\cite{boost-accuracy} is given using 0.3-differential privacy
as query history to help infer queries.
Let $2\epsilon$ be uniformly distributed in $[50,10^3]$, $\delta=0.2$.

We improve the baseline interactive query-answering system with our budget allocation method(Theorem \ref{theo-budget-allocation}), then it can save privacy budget and achieve better utility.
The improved baseline systems are compared with or without the Bayesian inference technique.
Because we don't allocate privacy budget if the estimate can satisfy utility requirement, privacy budget can be saved furthermore.
In Figure \ref{fig-impact-inference}, $R_i$ and $E$ also become better using inference technique.


%

\vspace{-0.2cm}
\begin{figure}
\begin{minipage}{0.015\textwidth}\vspace{-3cm}$\bar{\alpha}$\end{minipage}
\includegraphics[height=3cm]{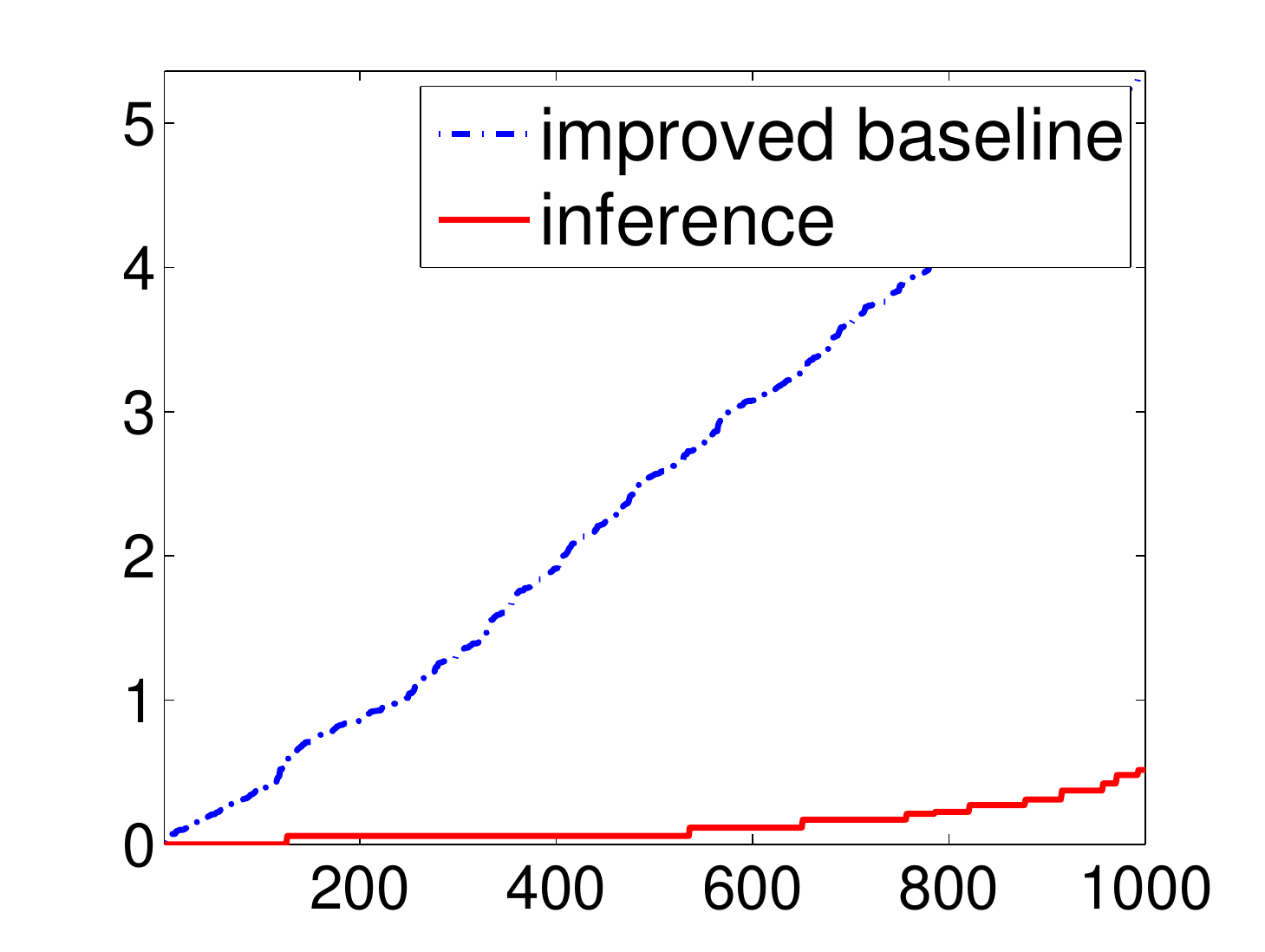}
\includegraphics[height=3cm]{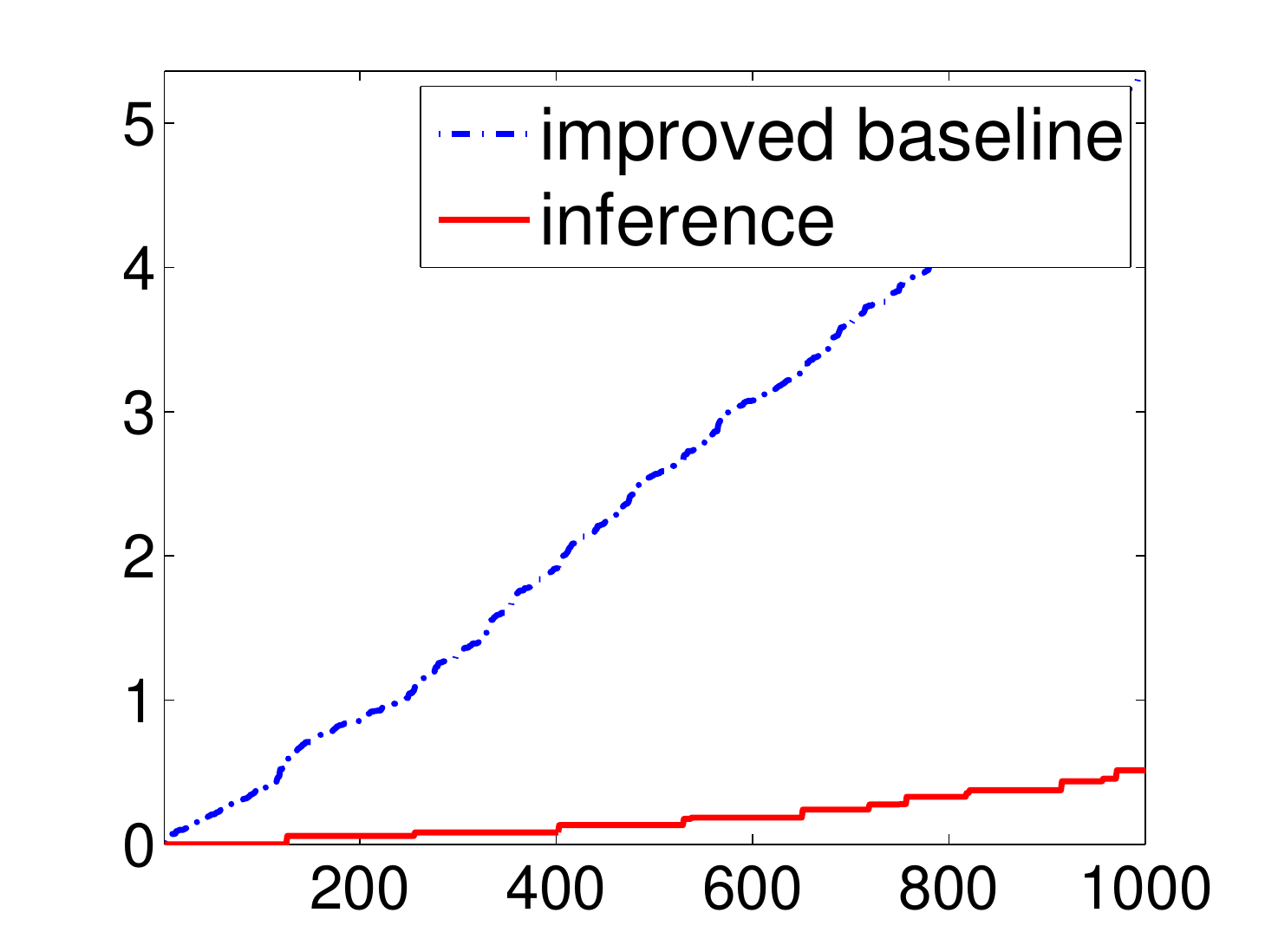}
\begin{minipage}{0.015\textwidth}\vspace{-3cm}$E$\end{minipage}
\includegraphics[height=3cm]{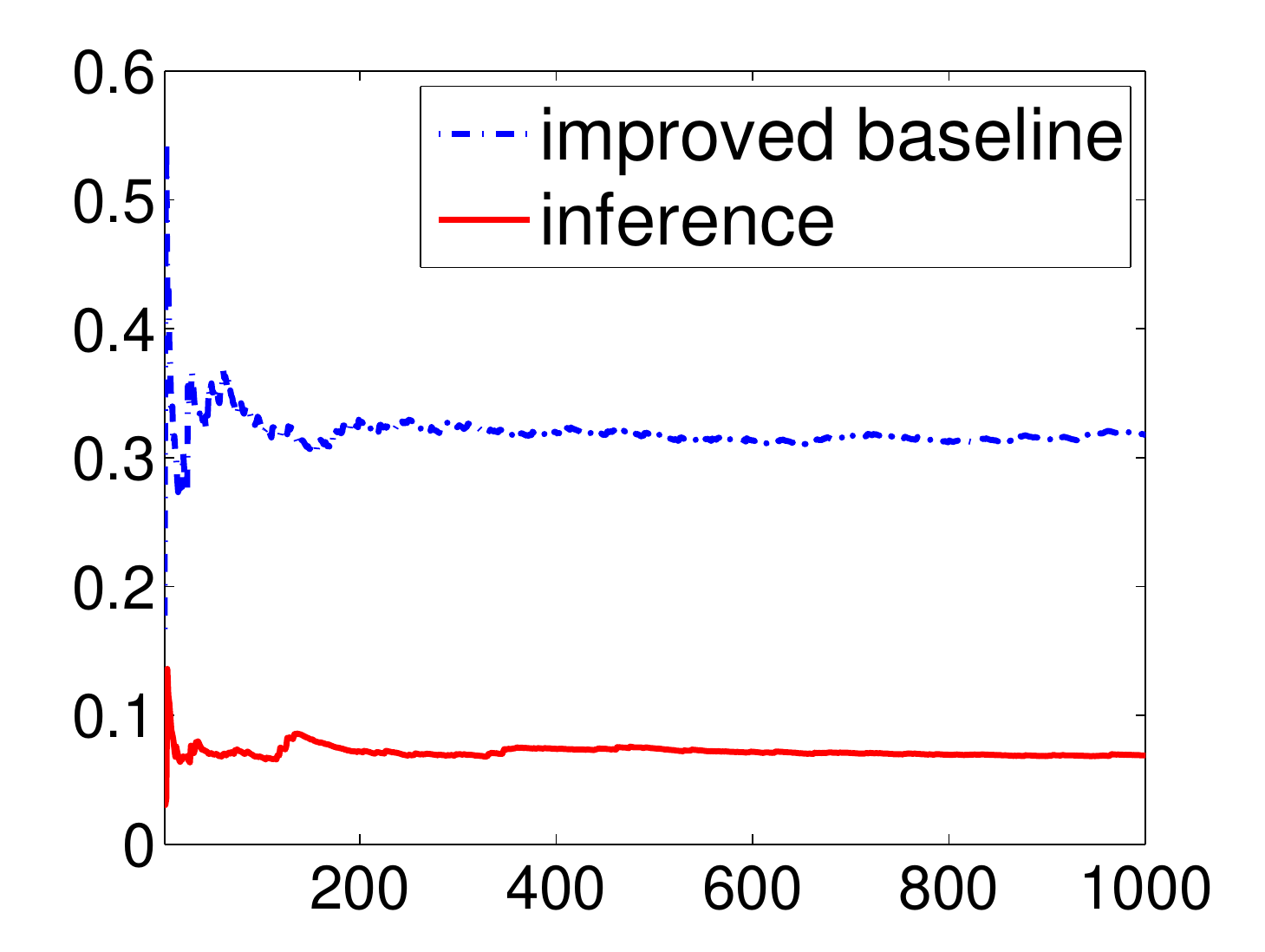}
\includegraphics[height=3cm]{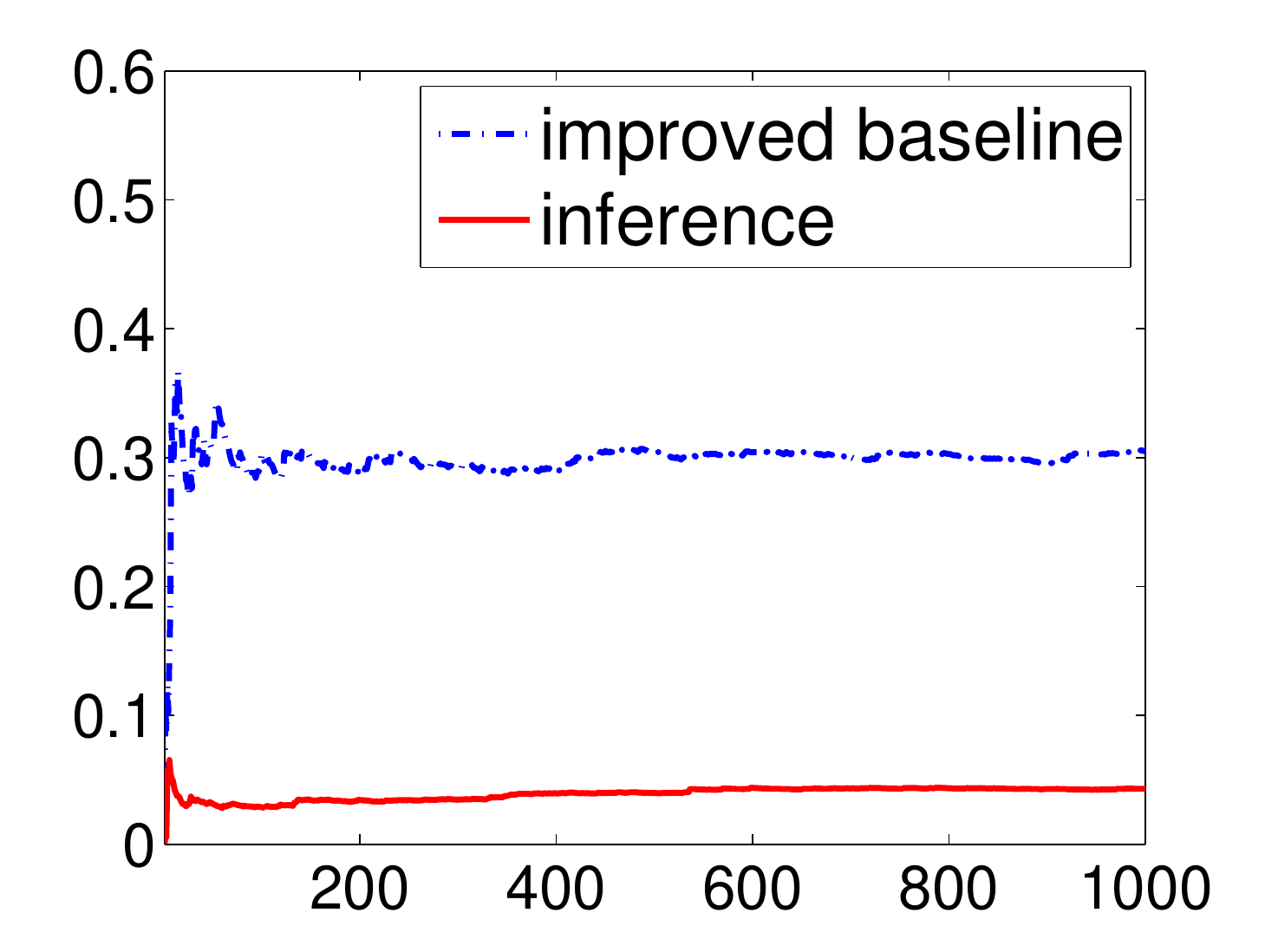}
\begin{minipage}{0.015\textwidth}\vspace{-3cm}$R_i$\end{minipage}
\includegraphics[height=3cm]{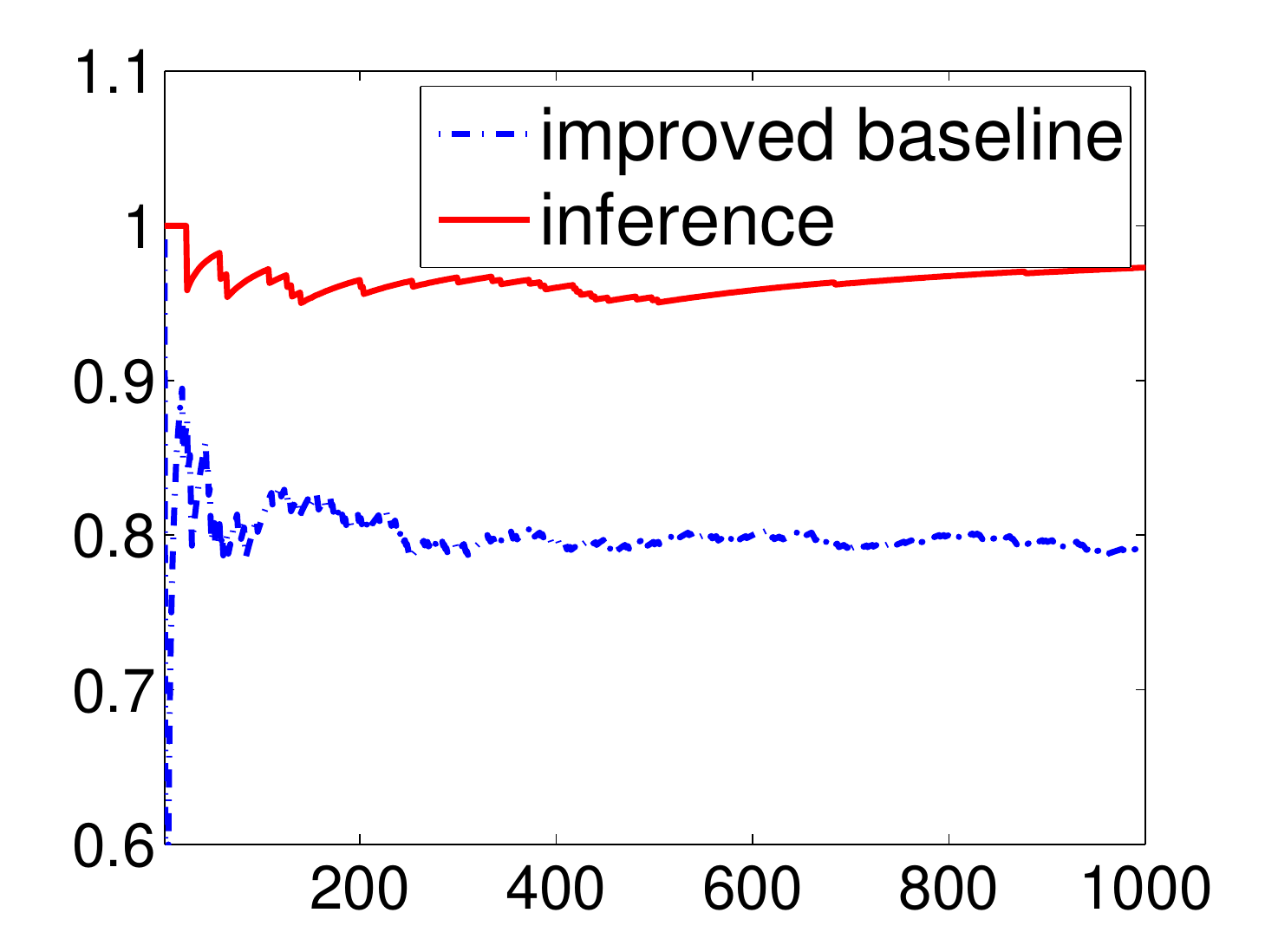}
\includegraphics[height=3cm]{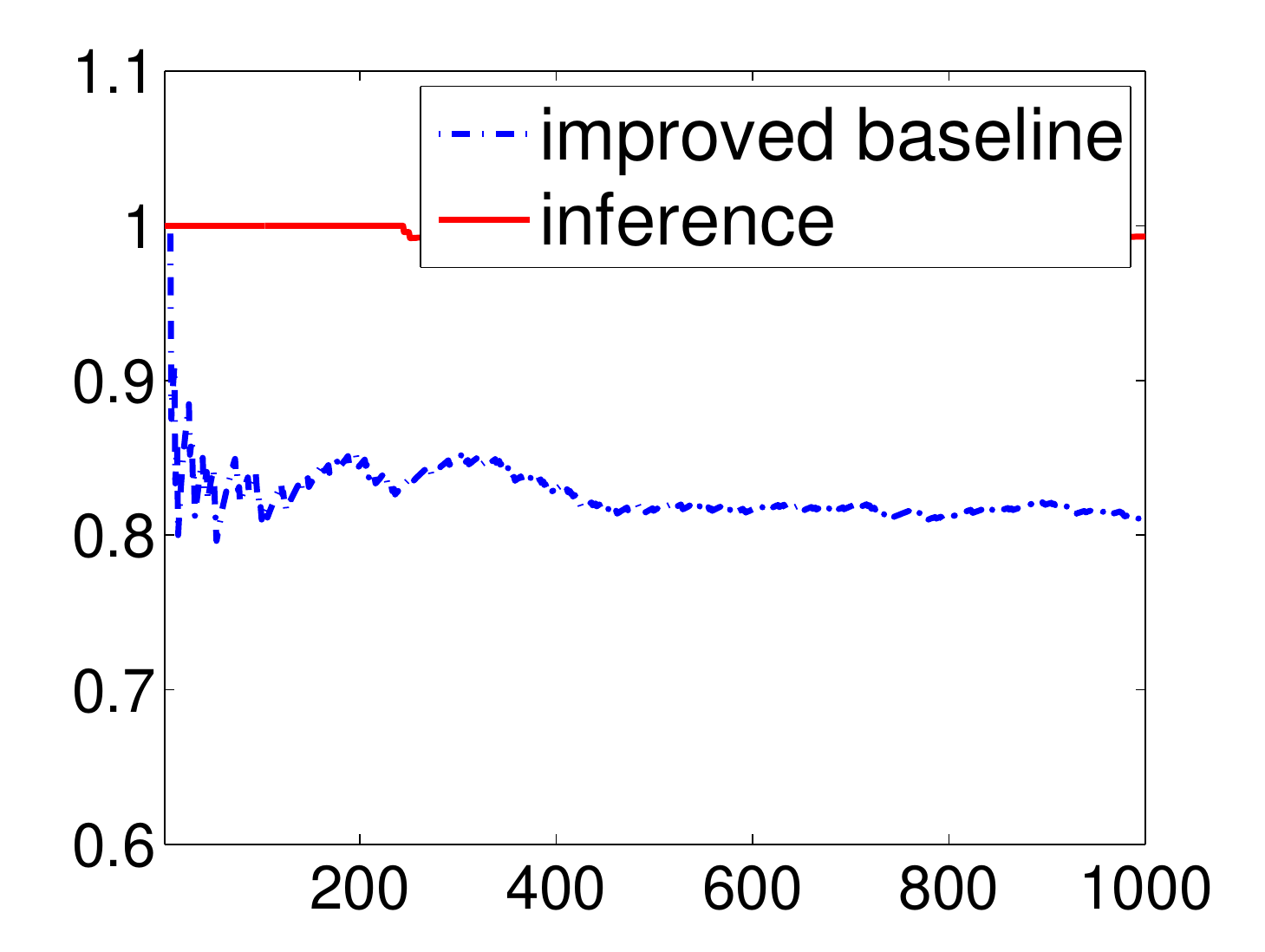}
\vspace{-0.7cm}
\caption{Evaluation of unbounded setting for net-trace(left) and search logs(right). All x-axes represent $\#$ of queries.
First row shows $\bar{\alpha}$; second row shows $E$; third row shows $R_i$. Improved baseline means baseline system with dynamic budget allocation.}
\label{fig-impact-inference}
\end{figure}
\vspace{-0.2cm}
\begin{figure}
\begin{minipage}{0.015\textwidth}\vspace{-3cm}$R_a$\end{minipage}
\includegraphics[height=3cm]{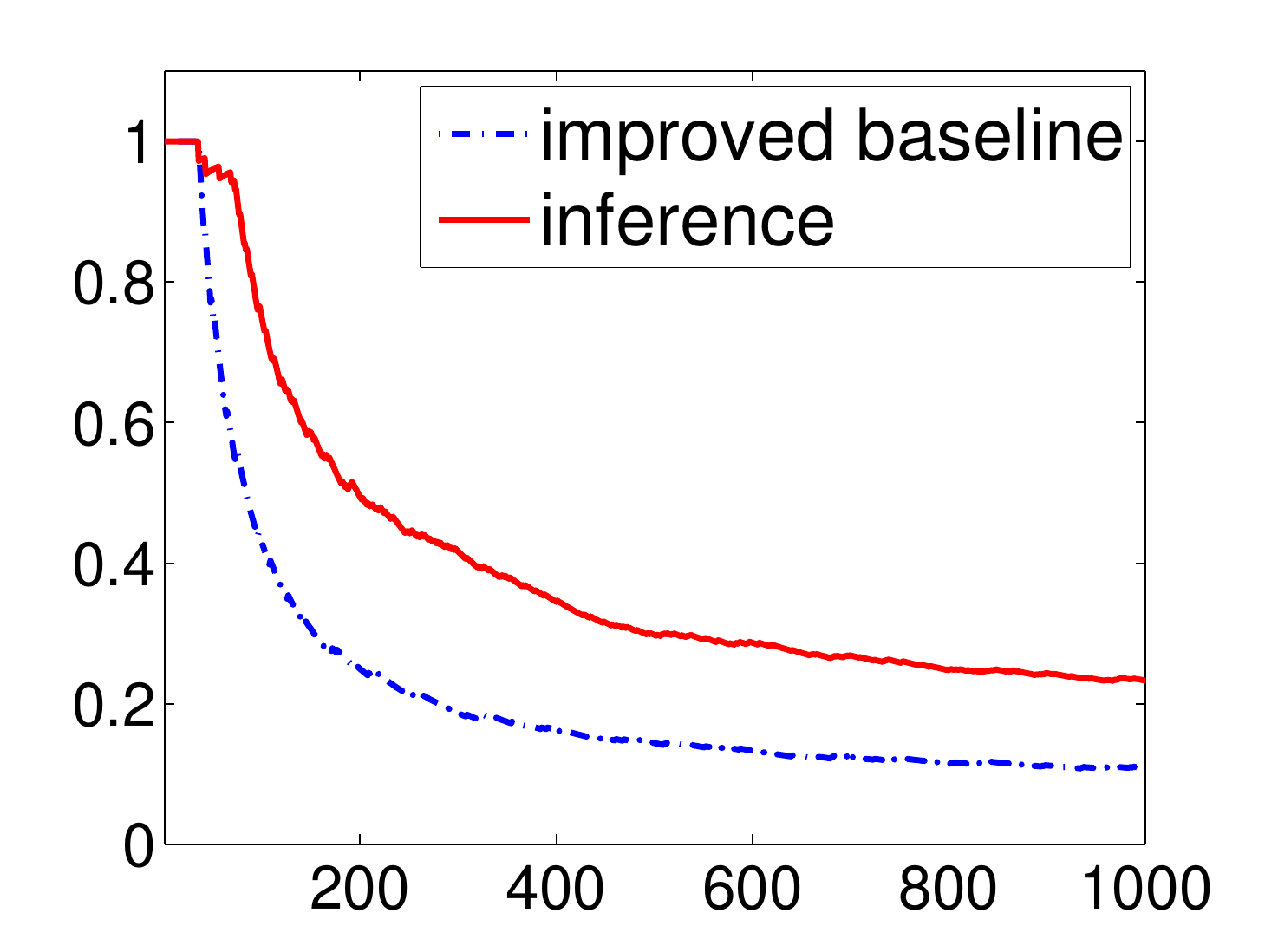}
\includegraphics[height=3cm]{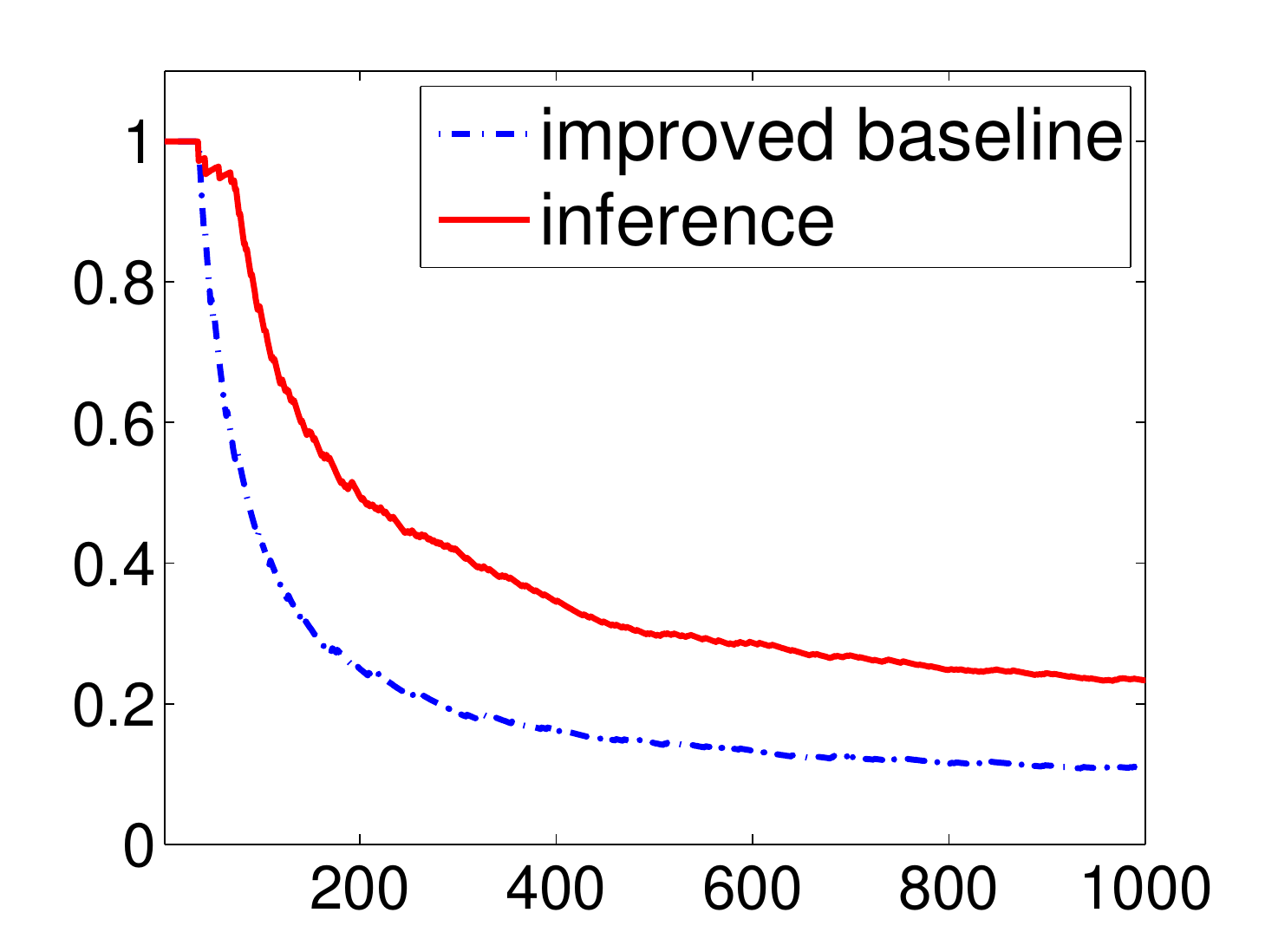}
\begin{minipage}{0.015\textwidth}\vspace{-3cm}$E$\end{minipage}
\includegraphics[height=3cm]{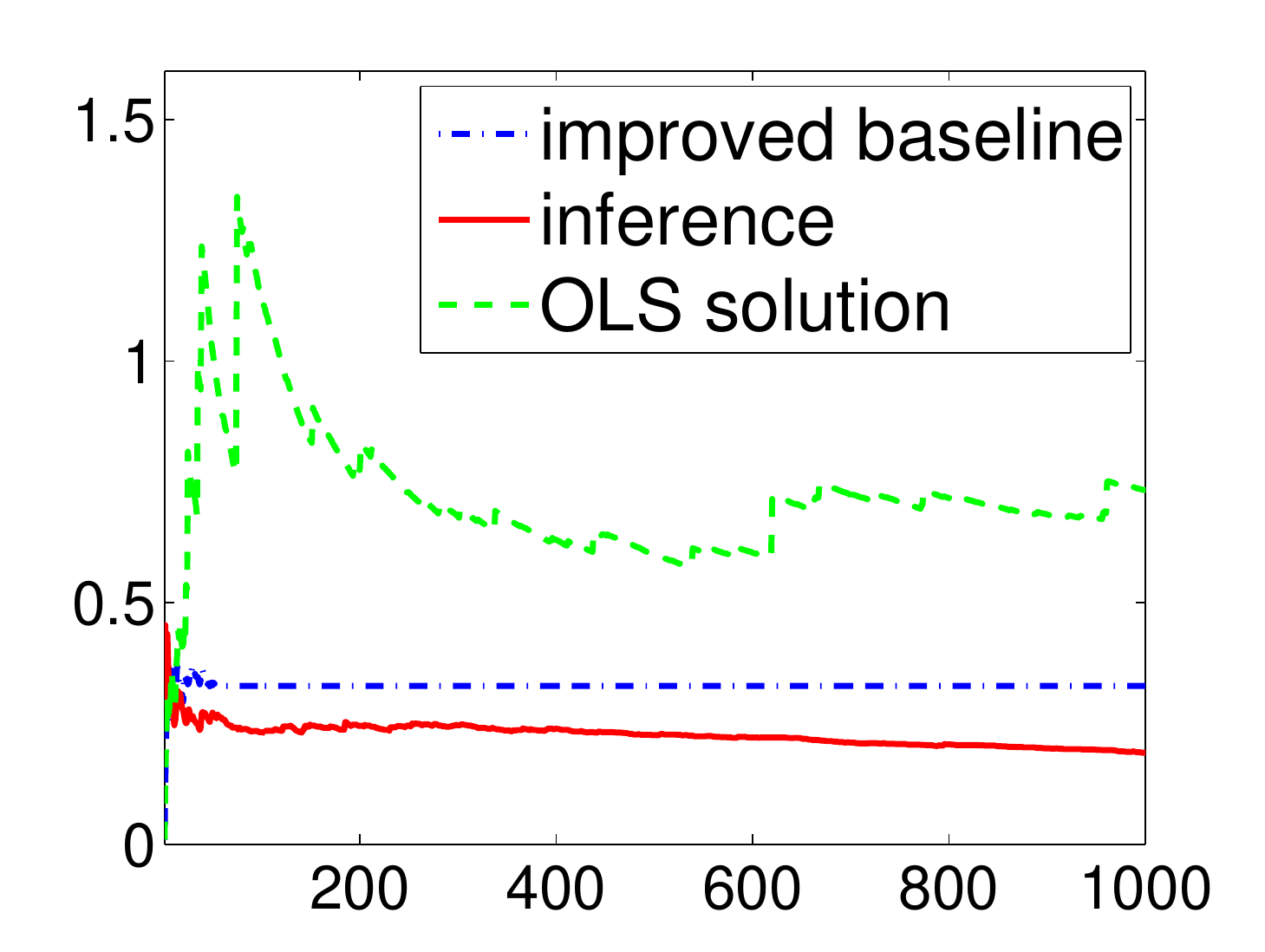}
\includegraphics[height=3cm]{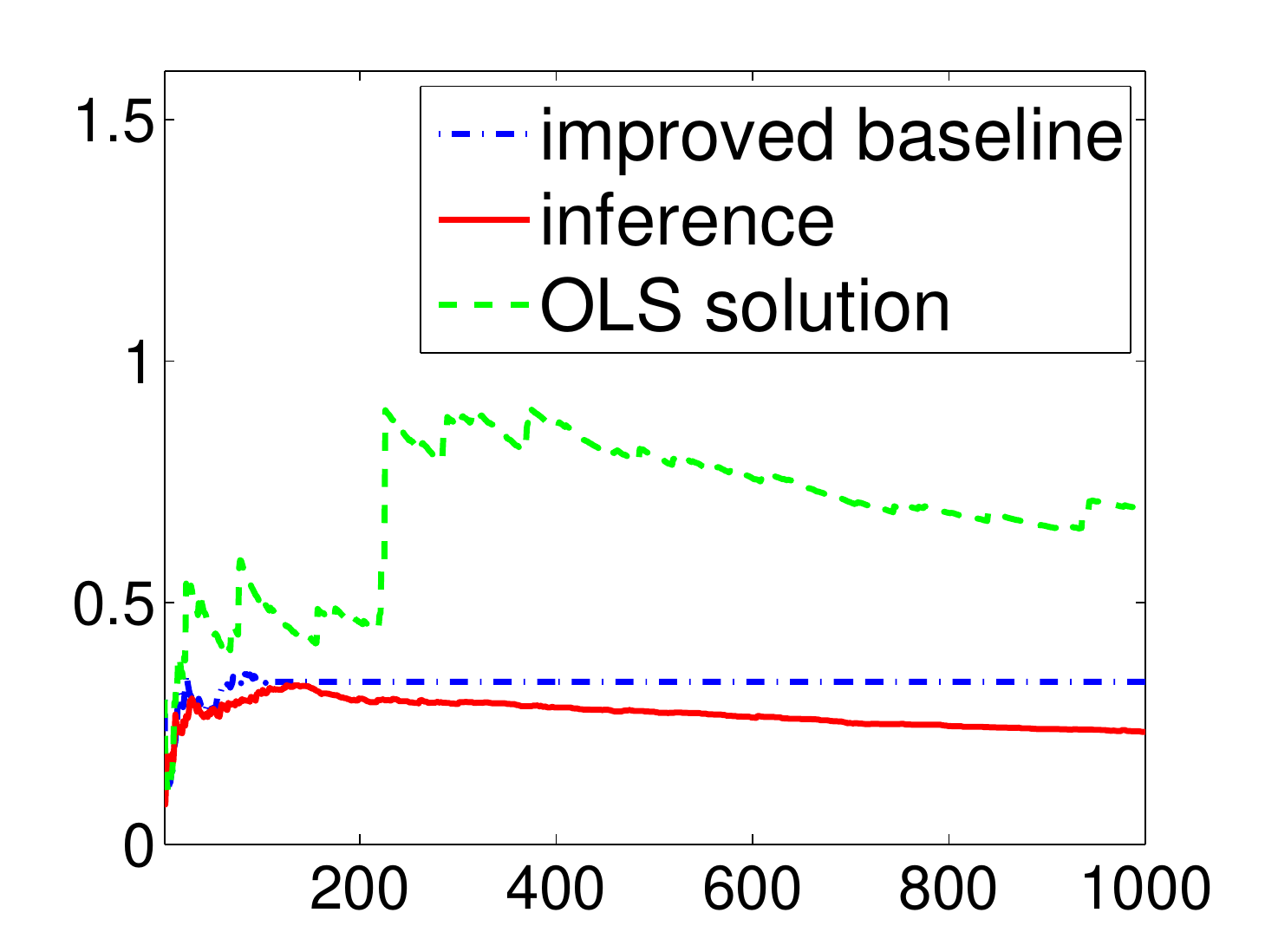}
\begin{minipage}{0.015\textwidth}\vspace{-3cm}$R_i$\end{minipage}
\includegraphics[height=3cm]{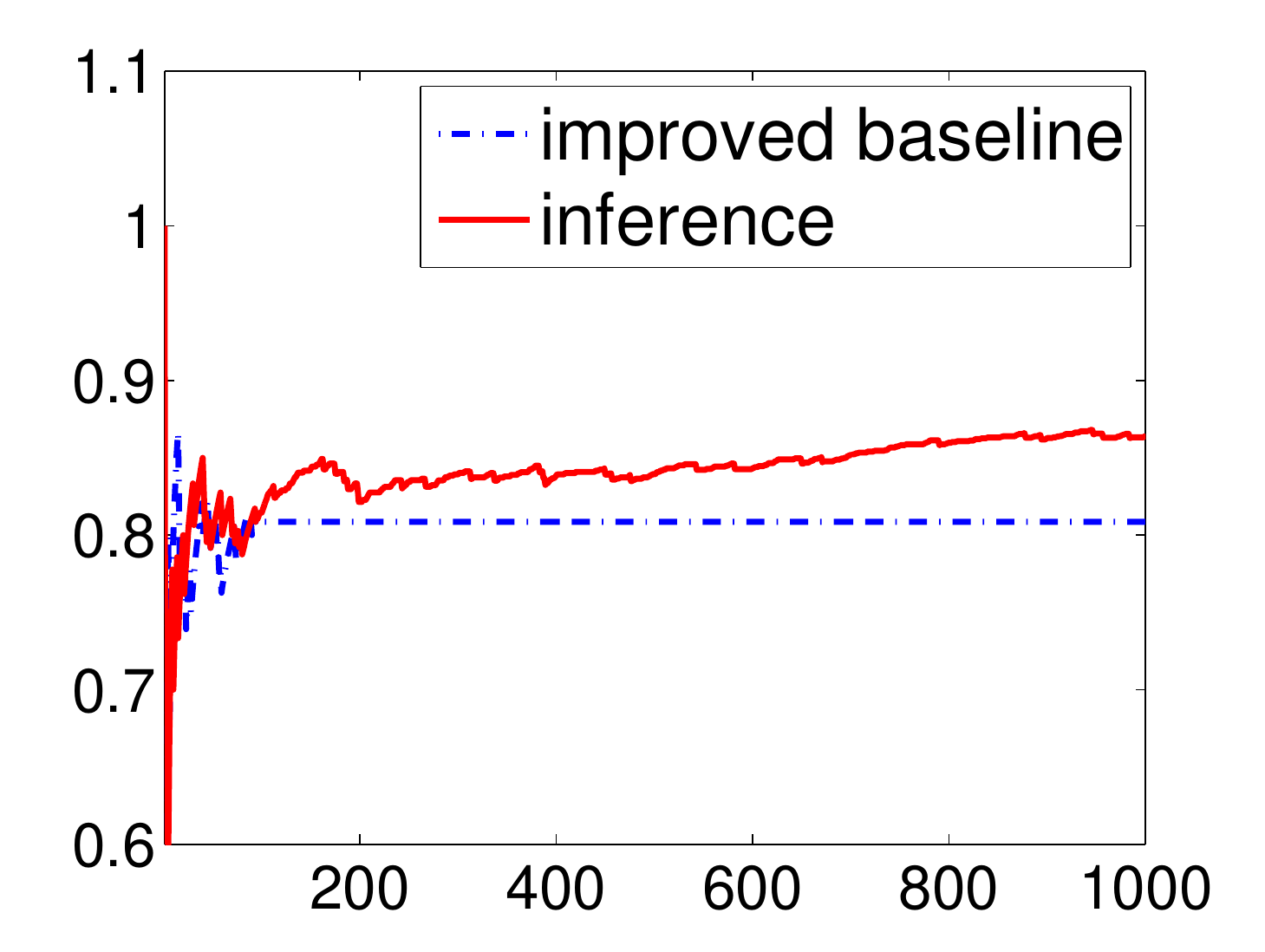}
\includegraphics[height=3cm]{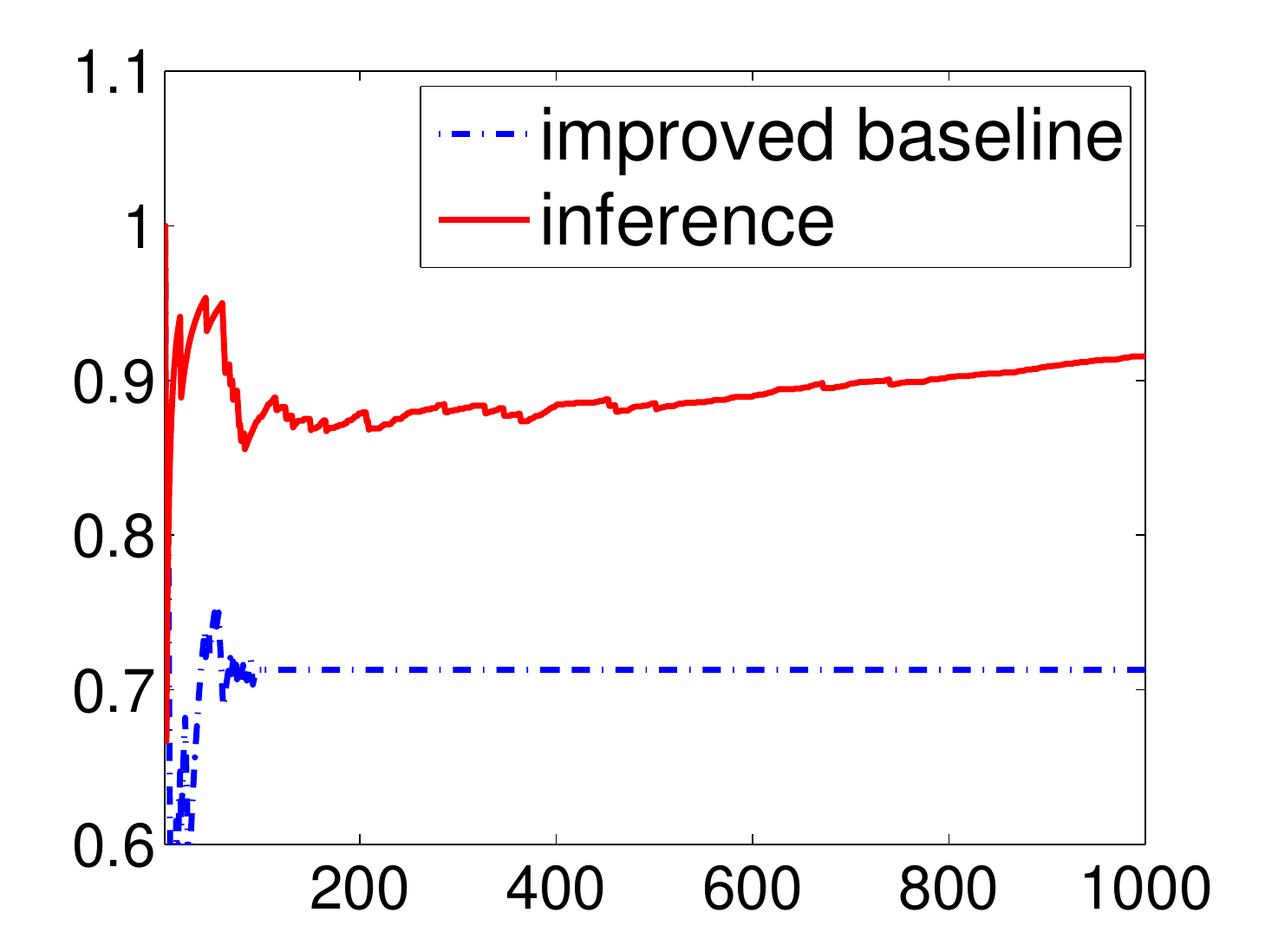}
\vspace{-0.7cm}
\caption{Evaluation of bounded setting for net-trace(left) and search logs(right). All x-axes represent $\#$ of queries.
First row shows $R_a$; second row shows $E$; third row shows $R_i$. Improved baseline and
OLS solution are improved by our method of dynamic budget allocation.
}
\label{fig-real_sys}
\end{figure}

\subsubsection{Bounded Setting}
\label{cha:exp-evaluation-real-system}
For a bounded system, we use both budget allocation and Bayesian inference to improve the performance.
Let the overall privacy budget be $1$;
let $2\epsilon$ be uniformly distributed in $[1,10^3]$, $\delta=0.2$. Figure \ref{fig-real_sys} shows the performance.
%
%

Because $\bar{\alpha}$ reaches the overall privacy budget for about $100$ queries, further queries cannot be answered by improved baseline system. But for our inference system,
they may also be answered by making inference of history queries. This can increase $R_a$. Also $E$ and $R_i$ is better than improved baseline system.
We also compared $E$ of OLS solution\cite{Optimizing-PODS} with budget allocation method.
In our setting, the result is better than OLS solution because OLS is not sensitive to
utility requirement and privacy cost $\boldsymbol\alpha$. For example, if a query requires high utility, say $2\epsilon=5$, then a big $\alpha$ is
allocated and a small noise is added to the answer. But OLS solution treats all noisy answers as points without any probability properties.
Therefore, this answer with small noise contributes not so much to the OLS estimate. Then according to Equation (\ref{eqn-relative-error}),
$E$ becomes lager than the inference result.

Note that $E$ and $R_i$ is not as good as Figure \ref{fig-impact-inference}. There are two reasons. First, we don't have any
query history at the beginning. So the system must build a query history to make inference for new queries. Second,
when $\bar{\alpha}$ reaches the overall privacy budget, the query history we have is actually the previous answered queries,
which are generated randomly. Compared with the hierarchical partitioning tree\cite{boost-accuracy}, it may not work so well.
We believe delicately designed query history\cite{SDMpaper,Optimizing-PODS,boost-accuracy} or auxiliary queries\cite{iReduct}
can improve our system even better, which can be investigated in further works.

\begin{small}
\bibliographystyle{abbrv}
\bibliography{ref/exportlist,ref/privacy}
\end{small}
\section{Appendix}
\label{app-notations}
\begin{table}[htbp]
\label{tbl-notations}
\begin{tabular}{p{45pt}p{170pt}}
\hline
{\textbf{H}} &matrix of query history;\\
$\textbf{x}$ & original count of each cell;\\
$\textbf{y}$ & noisy answer of $\textbf{H}$, $\textbf{y}=\textbf{Hx}+\tilde{\textbf{N}}$;\\
$\alpha$ & privacy budget for one query;\\
$\bar{\alpha}$ & privacy cost for the whole system;\\
$L$ and $U$ & upper and lower bound of credible interval;\\
$\epsilon$ & $\epsilon=(U-L)/2$;\\
$\delta$ & $1-\delta$ is the confidence level;\\
$\boldsymbol\alpha$& privacy cost of each row in $\textbf{H}$;\\
$\textbf{S}$& the sensitivity of each row in $\textbf{H}$;\\
$\theta$ or Q(D)& original answer of query $Q$;\\
$\tilde{N}(\alpha)$& Laplace noise with parameter $\alpha$;\\
$\mathcal{A}_Q(D)$& returned answer of Laplace mechanism: $\tilde{\theta}=\theta+\tilde{N}(\alpha/S)$;\\
$\tilde{\theta}$ & observation of a query: $\tilde{\theta}=\theta+N$;\\
$\hat{\theta}$ & estimated answer of $Q(D)$;\\
$\tilde{\textbf{N}}$ & noise vector of $\textbf{y}$;\\
\hline
\end{tabular}
\end{table}

%
%
%

\end{document}